\documentclass[11pt]{article}
\usepackage{fullpage} % smaller margins

\usepackage[all=normal,bibliography=tight,bibnotes=tight]{savetrees}
 \usepackage{graphicx} 
 \usepackage{xspace}
 \usepackage{url}
  \usepackage{amsmath,amssymb,amstext,amsthm}
\usepackage{graphicx}

 \usepackage{graphicx} 
 \usepackage{xspace}
 \usepackage{url}
 \usepackage{comment}
  \usepackage{amsmath,amssymb,amstext,amsthm}
\usepackage{graphicx}

  \newtheorem{theorem}{Theorem}%[section]
  \newtheorem{lemma}[theorem]{Lemma}
  \newtheorem{proposition}[theorem]{Proposition}
  \newtheorem{definition}[theorem]{Definition}

  \newtheorem{corollary}{Corollary}

\newcommand{\cO}{\mathcal{O}}
\newcommand{\sm}{\setminus}

\newcommand{\ham}{\mathcal{H}}

\newcommand{\clustering}{\textsc{Cluster Editing}\xspace}
\newcommand{\pclustering}{\textsc{$p$-Cluster Editing}\xspace}
\newcommand{\plclustering}{\textsc{$p_\leq$-Cluster Editing}\xspace}

\title{Subexponential fixed-parameter tractability of cluster editing\footnote{The authors from the University of Bergen are supported by the European Research Council (ERC) via grant Rigorous Theory of Preprocessing, reference 267959, and by the Research Council of Norway. This work has been done while the second author was at the Utrecht University, the Netherlands, and supported by the Dutch Research Foundation (NWO). The third author is supported by the National Science Centre grant N206 567140 and the Foundation for Polish Science.}}

\author{Fedor V. Fomin\thanks{Department of Informatics, University of Bergen, Norway,\newline   \texttt{\{fomin,michal.pilipczuk,yngve.villanger\}@ii.uib.no}}
\and Stefan Kratsch\thanks{Technical University Berlin, Germany, \texttt{stefan.kratsch@tu-berlin.de}}
\and Marcin  Pilipczuk\thanks{Institute of Informatics, University of Warsaw, Poland, \texttt{malcin@mimuw.edu.pl }\newline }  \addtocounter{footnote}{-3}
\and Micha{\l}  Pilipczuk\footnotemark  \addtocounter{footnote}{-1}
\and Yngve Villanger\footnotemark
}

\begin{document}

\date{}
\maketitle

 \begin{abstract}
In the \textsc{Correlation Clustering} problem, also known as \textsc{Cluster Editing}, we are given an undirected graph $G$ and a positive integer $k$; the task is to decide whether $G$ can be transformed into a cluster graph, i.e., a disjoint union of cliques, by changing at most $k$ adjacencies, that is, by adding or deleting at most $k$ edges. 

%The motivation of the problem stems from various tasks in computational biology (Ben-Dor et al., Journal of Computational Biology 1999) and machine learning (Bansal et al., Machine Learning 2004). Although in general \textsc{Correlation Clustering} is APX-hard (Charikar et al., FOCS 2003), the version of the problem where the number of cliques may not exceed a prescribed constant $p$ admits a PTAS (Giotis and Guruswami, SODA 2006).

We study the parameterized complexity of \textsc{Correlation Clustering} with this restriction on the number of cliques to be created. We give an algorithm that 
\begin{itemize}
\item in time $\cO(2^{\cO(\sqrt{pk})} + n+m)$ decides whether a graph $G$ on $n$ vertices and
$m$ edges can be transformed into a cluster graph with exactly $p$ cliques by changing at most $k$ adjacencies. 
\end{itemize}

We complement these algorithmic findings by the following, surprisingly tight lower bound on the asymptotic behavior of our algorithm. We show that unless the Exponential Time Hypothesis (ETH) fails 
\begin{itemize}
\item  
for any constant $0 \leq \sigma \leq 1$, there is $p = \Theta(k^\sigma)$ such that there is no algorithm deciding in time 
$2^{o(\sqrt{pk})} \cdot  n^{\cO(1)}$ whether an $n$-vertex graph $G$ can be transformed into a cluster graph with at most $p$ cliques by changing at most $k$ adjacencies.
 \end{itemize}
Thus, our upper and lower bounds provide an asymptotically tight analysis of the multivariate parameterized complexity of the problem for the whole range of values of $p$ from constant to a linear function of $k$.
 \end{abstract}
 
\section{Introduction}\label{sec:intro}
\emph{Correlation clustering}, also known as \emph{clustering with qualitative information} or \emph{cluster editing}, is the problem to cluster objects based only on 
the qualitative information concerning similarity between pairs of them. For every pair of objects we have a binary indication whether they are similar or not. The task is to find a partition of the objects into clusters minimizing the number of similarities between different clusters and non-similarities inside of clusters. The problem was introduced by Ben-Dor, Shamir, and Yakhini \cite{Ben-DorSY99} motivated by problems from computational biology, and, independently, by Bansal, Blum, and Chawla \cite{Bansal04}, motivated by machine learning problems concerning document clustering according to similarities. The correlation version of clustering was studied intensively, including \cite{AilonCN08,AlonMMN05,AroraBKSH05,CharikarGW05j,CharikarW04,GiotisG06,ShamirST04}.
%\cite{AilonCN08,AlonMMN05,CharikarW04,GiotisG06,ShamirST04}.

The graph-theoretic formulation of the problem is the following. A graph $K$ is a \emph{cluster graph} if every connected component of $K$ is a complete graph.
Let $G=(V,E)$ be a graph; then $F \subseteq V \times V$ is called a \emph{cluster editing set} for $G$ if $G\triangle F =(V,E\triangle F)$ is a cluster graph. Here $E\triangle F$ is the symmetric difference between $E$ and $F$. In the optimization version of the problem the task is to find a cluster editing set of minimum size. Constant factor approximation algorithms for this problem were obtained in  \cite{AilonCN08,Bansal04,CharikarGW05j}. On the negative side, the problem is known to be NP-complete \cite{ShamirST04} and, as was shown by Charikar, Guruswami, and Wirth~\cite{CharikarGW05j}, also APX-hard.

Giotis and Guruswami \cite{GiotisG06} initiated the study of clustering when the maximum number of clusters that we are allowed to use is stipulated to be a fixed constant $p$. As observed by them, this type of clustering is  well-motivated in settings where the number of clusters might be an external constraint that has to be met.  It appeared that $p$-clustering variants posed new and non-trivial challenges. In particular, in spite of the APX-hardness of the general case, Giotis and Guruswami \cite{GiotisG06} gave a PTAS for this version of the problem. 

A cluster graph $G$  is called a \emph{$p$-cluster graph}  if it has exactly $p$ connected components or, equivalently, if it is a disjoint union of exactly $p$ cliques. 
Similarly, a set $F$ is a \emph{$p$-cluster editing set} of $G$, if $G\triangle F$ is a $p$-cluster graph.
In parameterized complexity, correlation clustering and its restriction to bounded number of clusters were studied under the names \textsc{Cluster Editing} and $p$-\textsc{Cluster Editing}, respectively.

\vspace{-0.1cm}
\begin{center}
\fbox{\begin{minipage}{13cm}
\noindent  \textsc{Cluster Editing} \hfill {\sl Parameter:} $k$.\\
{\sl Input:} A graph $G =(V,E) $ and a non-negative integer $k$.\\
% {\sl Parameter:} $k$.\\ 
{\sl Question:} Is there a cluster editing set for $G$ of size at most $k$? 
%$G_i$\\
\end{minipage}}
\end{center}
\vspace{-0.5cm}
 
\begin{center}
\fbox{\begin{minipage}{13cm}
\noindent  \textsc{$p$-Cluster Editing} \hfill {\sl Parameters:} $p,k$.\\
{\sl Input:} A graph $G =(V,E) $ and non-negative integers $p$ and $k$.\\
% {\sl Parameter:} $k$.\\ 
{\sl Question:} Is there a $p$-cluster editing set for $G$ of size at most $k$?
\end{minipage}}
\end{center}
\vspace{-0.1cm}

The parameterized version of \textsc{Cluster Editing}, and variants of it, were studied  intensively \cite{bocker:iwoca,BockerBBT08,BockerBK11,BockerD11,BodlaenderFHMPR10,Damaschke10,FellowsGKNU11,GrammGHN05,GuoKKU11,GuoKNU10,komusiewicz:sofsem,ProttiSS09}. The problem is solvable in time 
$\cO(1.62^k +n+m)$ \cite{bocker:iwoca} and it has a kernel with $2k$ vertices
\cite{CaoC10,ChenM10} (see Section~\ref{sec:prelim} for the definition of a kernel).    
  Shamir et al.~\cite{ShamirST04} showed that \textsc{$p$-Cluster Editing} is NP-complete for every fixed $p\geq 2$. A kernel with $(p+2)k+p$ vertices was given by Guo \cite{Guo09}. 
 
\paragraph*{Our results} We study the impact of the interaction between $p$ and $k$ on the parameterized complexity of \textsc{$p$-Cluster Editing}. Our main algorithmic result is the following.

\begin{theorem}\label{thm:pclustering-subept}
\pclustering{} is solvable in time $\cO(2^{\cO(\sqrt{pk})} + {m+n})$.
\end{theorem}

It is straightforward to modify our algorithm to work also in the following variants of the problem, where each edge and non-edge is assigned some edition cost: either {\emph{(i)}} all costs are at least one and $k$ is the bound on the maximum total cost of the solution, or {\emph{(ii)}} we ask for a set of at most $k$ edits of minimum cost. Let us also remark that, by Theorem~\ref{thm:pclustering-subept}, if $p=o(k)$ then \pclustering{} can be solved in $2^{o(k)}n^{\cO(1)}$ time, and thus it belongs to complexity class SUBEPT defined by Flum and Grohe \cite[Chapter~16]{FlumGrohebook}. Until very recently,  the only problems known to be in the class  SUBEPT were  the  problems with additional constraints on the input, like being a  planar,  $H$-minor-free, or tournament  graph \cite{AlonLS09,DemaineFHT05jacm}. However, recent algorithmic developments indicate  that the structure of the class  SUBEPT is much more interesting than expected. It appears that some parameterized problems related to chordal graphs, like  \textsc{Minimum Fill-in} or \textsc{Chordal Graph Sandwich}, are also in SUBEPT \cite{FominV126soda}. 

We would like to remark that \pclustering{} can be also solved in worse time complexity $\cO((pk)^{\cO(\sqrt{pk})} + {m+n})$ using simple guessing arguments. 
One such algorithm is based on the following observation:
Suppose that, for some integer $r$, we know at least $2r+1$ vertices from each cluster. Then, if an unassigned vertex has at most $r$
incident modifications, we know precisely to which cluster it belongs: it is adjacent to at least $r+1$ vertices already assigned to its cluster
and at most $r$ assigned to any other cluster. On the other hand, there are at most $2k/r$ vertices with more than $r$ incident modifications.
Thus (i) guessing $2r+1$ vertices from each cluster (or all of them, if there are less than $2r+1$), and (ii) guessing all vertices with more than $r$ incident modifications, together
with their alignment to clusters, results in at most $n^{(2r+1)p} n^{2k/r} p^{2k/r}$ subcases.
By pipelining it with the kernelization of Guo~\cite{Guo09} and with simple reduction rules that ensure $p \leq 6k$
(see Section~\ref{sec:large-p} for details), we obtain the claimed time complexity
for $r \sim \sqrt{k/p}$.
%$\cO(2^{\cO(rp \log(pk))} 2^{\cO(2k/r \cdot \log(pk))} + n +m)$, which gives $\cO((pk)^{\cO(\sqrt{pk})} + n + m)$ for $r \sim \sqrt{k/p}$.

An approach via {\emph{chromatic coding}}, introduced by Alon et al.~\cite{AlonLS09}, also leads to an algorithm with running time $\cO(2^{\cO(p\sqrt{k}\log p)}+n+m)$. However, one needs to develop new concepts to construct an algorithm for \pclustering{} with complexity bound as promised in Theorem~\ref{thm:pclustering-subept}, and thus obtain a subexponential complexity for every sublinear $p$.

The crucial observation is that a $p$-cluster graph, for $p=\cO(k)$, has $2^{\cO(\sqrt{pk})}$ edge cuts of size at most $k$ (henceforth called {\em{$k$-cuts}}). As in a YES-instance to the \pclustering{} problem each $k$-cut is a $2k$-cut of a $p$-cluster graph, we infer a similar bound on the number of cuts if we are dealing with a YES-instance. This allows us to use dynamic programming over the set of $k$-cuts. Pipelining this approach with a kernelization algorithm for \pclustering{} proves Theorem~\ref{thm:pclustering-subept}.

A new and active direction in parameterized complexity is the pursuit of asymptotically tight  bounds on the complexity of problems. In several cases, it is possible to obtain a complete analysis  by providing  matching lower (complexity) and upper (algorithmic) bounds. We refer to the recent survey of Marx~\cite{marx:future}, where  recent developments in the area are discussed, and the ``optimality program" is announced  among the main future research directions in parameterized complexity. 
 The most widely used complexity assumption for such tight lower bounds is the \emph{Exponential Time Hypothesis (ETH)}, which  posits that no subexponential-time algorithms for~$k${\sc{-CNF-SAT}} or {\sc{CNF-SAT}} exist~\cite{ImpagliazzoPZ01}.
   
Following this direction, we complement Theorem~\ref{thm:pclustering-subept} with two lower bounds. Our first, main lower bound is based on the following technical Theorem~\ref{thm:multivariate-reduction}, which shows that the exponential time dependence of our algorithm is asymptotically tight for any choice of parameters $p$ and $k$, where $p=\cO(k)$. As one can provide polynomial-time reduction rules that ensure that $p\leq 6k$ (see Section~\ref{sec:large-p} for details), this provides a full and tight picture of the multivariate parameterized complexity of \pclustering{}: we have asymptotically matching upper and lower bounds on the whole interval between $p$ being a constant and linear in $k$. To the best of our knowledge, this is the first fully multivariate and tight complexity analysis of a parameterized problem.
  
\begin{theorem}\label{thm:multivariate-reduction}
For any  $\varepsilon > 0 $ there is $\delta > 0$ and
a polynomial-time algorithm that, given positive integers $p$ and $k$
and a $3$-CNF-SAT formula $\Phi$ with $n$ variables and $m$ clauses, 
    such that $k,n \geq \varepsilon p$ and $n,m \leq \sqrt{pk}/\varepsilon$,
    computes a graph $G$ and integer $k'$,
    such that $k' \leq \delta k$, $|V(G)| \leq \delta\sqrt{pk}$, and
    \begin{itemize}
    \item if $\Phi$ is satisfiable then there is a $6p$-cluster graph $G_0$ with $V(G) = V(G_0)$ and $|E(G) \triangle E(G_0)| \leq k'$; %and
    \item if there exists a $p'$-cluster graph $G_0$ with $p' \leq 6p$, $V(G) = V(G_0)$ and $|E(G) \triangle E(G_0)| \leq k'$, then $\Phi$ is satisfiable.
    \end{itemize}
\end{theorem}

As the statement of Theorem \ref{thm:multivariate-reduction} may look technical, we gather its two main consequences in Corollaries \ref{cor:p-k} and \ref{cor:fixed-p}. We state both corollaries in terms of an easier \plclustering problem, where the number of clusters has to be at most $p$ instead of precisely equal to $p$. Clearly, this version can be solved by an algorithm for \textsc{$p$-Cluster Editing} with an additional $p$ overhead in time complexity by trying all possible $p'\leq p$, so the lower bound holds also for harder \pclustering; however, we are not aware of any reduction in the opposite direction. In both corollaries we use the fact that existence of a subexponential, in both the number of variables and clauses, algorithm for verifying satisfiability of $3$-CNF-SAT formulas would violate ETH \cite{ImpagliazzoPZ01}.

\begin{corollary}\label{cor:p-k}
Unless ETH fails, for every $0 \leq \sigma \leq 1$, there is $p = \Theta(k^\sigma)$ such that \plclustering is not solvable in time $2^{o(\sqrt{pk})} |V(G)|^{\cO(1)}$.
\end{corollary}
\begin{proof}
Assume we are given a $3$-CNF-SAT formula $\Phi$ with $n$ variables and $m$ clauses.
If $n < m$, then $\lceil (m-n)/2 \rceil$ times perform the following operation: add three new variables $x$, $y$ and $z$, and clause $(x \vee y \vee z)$. In this way we preserve the satisfiability of $\Phi$, increase the size of $\Phi$ by a constant factor, and ensure that $n \geq m$.

Take now $k = \lceil n^\frac{2}{1+\sigma} \rceil$, $p = \lceil n^\frac{2\sigma}{1+\sigma} \rceil$. As $n \geq m$ and $0 < \sigma \leq 1$,
we have $k,n \geq p$ and $n,m \leq \sqrt{pk}$ but $n+m = \Omega(\sqrt{pk})$. Invoke Theorem \ref{thm:multivariate-reduction}
for $\varepsilon = 1$ and apply the reduction algorithm for the formula $\Phi$ and parameters $p$ and $k$,
obtaining a graph $G$ and a parameter $k'$. Note that $6p = \Theta(k^\sigma)$.
Apply the assumed algorithm for the \plclustering problem to the instance $(G,6p,k')$.
In this way we resolve the satisfiability of $\Phi$ in time
$2^{o(\sqrt{pk})}|V(G)|^{\cO(1)} =  2^{o(n+m)}$, contradicting ETH.
\end{proof}

\begin{corollary}\label{cor:fixed-p}
Unless ETH fails, for every constant $p \geq 6$, there is no algorithm solving \plclustering in time
$ 2^{o(\sqrt{k})} |V(G)|^{\cO(1)}$ or $2^{o(|V(G)|)}$.
\end{corollary}
\begin{proof}
We prove the corollary for $p=6$; the claim for larger values of $p$ can be proved easily taking the graph obtained in the reduction and introducing additional $p-6$ cliques of its size.

Assume we are given a $3$-CNF-SAT formula $\Phi$ with $n$ variables and $m$ clauses.
Take $k = \max(n,m)^2$, invoke Theorem \ref{thm:multivariate-reduction} for $\varepsilon = 1$
and feed the reduction algorithm with the formula $\Phi$ and parameters $1$ and $k$, obtaining a graph $G$ and a parameter $k'$.
Note that $|V(G)| = \cO(\sqrt{k})$.
Apply the assumed algorithm for the \plclustering problem to the instance $(G,6,k')$.
In this way we resolve the satisfiability of $\Phi$ in time
$ 2^{o(\sqrt{k})}|V(G)|^{\cO(1)}=  2^{o(n+m)}$, contradicting ETH.
\end{proof}

Note that Theorem~\ref{thm:multivariate-reduction} and {Corollary}~\ref{cor:p-k} do not rule out possibility that the general \textsc{Cluster Editing} is solvable in subexponential time.  
Our second, complementary lower bound shows that when the number of clusters is not constrained, then the problem cannot be solved in subexponential time unless ETH fails. This disproves the conjecture of Cao and Chen \cite{CaoC10}. We note that Theorem~\ref{thm:eth} was independently obtained by Komusiewicz in his PhD thesis
\cite{komusiewicz:thesis}. %, using a reduction very similar to ours.

\begin{theorem}\label{thm:eth}
Unless ETH fails, \textsc{Cluster Editing} cannot be solved in time $2^{o(k)}n^{\cO(1)}$.
\end{theorem}

Clearly, by Theorem~\ref{thm:pclustering-subept}, the reduction of Theorem~\ref{thm:eth} must produce an instance where the number of clusters in any solution, if there exists any, is $\Omega(k)$. Therefore, intuitively the hard instances of \textsc{Cluster Editing} are those where every cluster needs just a constant number of adjacent editions to be extracted.

\paragraph*{Organization of the paper} In Section~\ref{sec:prelim} we establish notation and recall classical notions and results that will be used throughout the paper. Section~\ref{sec:subept-new} contains description of the subexponential algorithm for \pclustering, i.e., the proof of Theorem~\ref{thm:pclustering-subept}. Section~\ref{sec:multi} is devoted to the multivariate lower bound, i.e., the proof of Theorem~\ref{thm:multivariate-reduction}, while in Section~\ref{app:eth} we give the lower bound for the general \clustering{} problem, i.e., the proof of Theorem~\ref{thm:eth}. In Section~\ref{sec:conclusions} we gather some concluding remarks and propositions for further work.

\section{Preliminaries}\label{sec:prelim}
We denote by $G=(V,E)$ a finite, undirected, and simple graph with vertex set $V(G)=V$ and edge set $E(G)=E$. We also use   $n$ to denote the number of vertices and $m$ the number of edges in $G$. For a nonempty subset $W \subseteq V$, the subgraph of $G$ induced by $W$ is denoted by $G[W]$. We say that a vertex set $W\subseteq V$ is \emph{connected} if $G[W]$ is connected. The \emph{open neighborhood} of a vertex $v$ is $N(v)=\{u\in V:~uv \in E\}$ and the \emph{closed neighborhood} is $N[v] = N(v) \cup \{v\}$. For a vertex set $W\subseteq V$ we put  $N(W) = \bigcup_{v \in W} N(v)\sm W$ and $N[W] = N(W) \cup W$.

For graphs $G,H$ with $V(G)=V(H)$, by $\ham(G,H)$ we denote the number of edge modifications needed to obtain $H$ from $G$, i.e.,  $\ham(G,H) = |E(G) \triangle E(H)|$.
By $E(X,Y)$ we denote the set of edges having one endpoint in $X$ and second in $Y$.

A parameterized problem $\Pi$ is a subset of $\Gamma^{*}\times \mathbb{N}$ for some finite alphabet $\Gamma$. An instance of a parameterized problem consists of $(x,k)$, where $k$ is called the parameter. 
 A central notion in 
parameterized complexity is {\em fixed-parameter tractability (FPT)} which means, for a given instance $(x,k)$, 
solvability in time $f(k)\cdot p(|x|)$, where $f$ is an arbitrary computable function of $k$ and $p$ is a polynomial in the input size. We refer to the book of Downey and Fellows~\cite{DowneyF99} for further reading on parameterized complexity. 

 A \emph{kernelization algorithm} for a  parameterized problem 
$\Pi\subseteq \Gamma^{*}\times \mathbb{N}$ is an algorithm that given $(x,k)\in \Gamma^{*}\times \mathbb{N} $ 
outputs in time polynomial in $|x|+k$ a pair $(x',k')\in \Gamma^{*}\times \mathbb{N}$, called a \emph{kernel} such that
$(x,k)\in \Pi$  if and only if $(x',k')\in \Pi$, 
 $|x'|\leq g(k)$,  and 
 $k' \leq k$, where $g$ is some computable function. 

In our algorithm we need the following result of Guo \cite{Guo09}.
 
\begin{proposition}[\cite{Guo09}] \label{prop:polykernel_Guo}
\textsc{$p$-Cluster Editing}  admits a kernel with $(p + 2)k + p$ vertices. The running time of the kernelization algorithm is $\cO(n+m)$, where $n$ is the number of vertices and $m$ the number of edges in the input graph $G$.
\end{proposition} 
 
The following lemma is used in both our lower bounds. Its proof
is almost identical to the proof of Lemma~1 in \cite{Guo09}, and we provide it here for reader's convenience. 

\begin{lemma}\label{lem:twins}
Let $G=(V,E)$ be an undirected graph and $X \subseteq V$ be a set of vertices such that $G[X]$ is a clique and
each vertex in $X$ has the same set of neighbors outside $X$ (i.e., $N_G[v] = N_G[X]$ for each $v \in X$).
Let $F \subseteq V \times V$ be a set such that $G \triangle F$ is a cluster graph where the vertices of $X$
are in at least two different clusters. Then there exists $F' \subseteq V \times V$ such that:
(i) $|F'| < |F|$, (ii) $G \triangle F'$ is a cluster graph with no larger number of clusters than $G \triangle F$,
(iii) in $G \triangle F'$ the clique $G[X]$ is contained in one cluster.
\end{lemma}
\begin{proof}
For a vertex $v \in X$, let $F(v) = \{u \notin X: vu \in F\}$.
Note that, since $N_G[v] = N_G[X]$ for all $v \in X$, we have $F(v) = F(w)$ if
$v$ and $w$ belong to the same cluster in $G \triangle F$.

Let $Z$ be the vertex set of a cluster in $G \triangle F$ such that
there exists $v \in Z \cap X$ with smallest $|F(v)|$. Construct $F'$ as follows:
take $F$, and for each $w \in X$ replace all elements of $F$ incident with $w$
with $\{uw: u \in F(v)\}$. In other words, we modify $F$ by moving all vertices
of $X \setminus Z$ to the cluster $Z$. Clearly, $G \triangle F'$ is a cluster graph,
   $X$ is contained in one cluster in $G \triangle F'$ and $G \triangle F'$
   contains no more clusters than $G \triangle F$.
To finish the proof, we need to show that $|F'| < |F|$.
The sets $F$ and $F'$ contain the same set of elements not incident with $X$.
As $|F(v)|$ was minimum possible, for each $w \in X$ we have $|F(w)| \geq |F'(w)|$.
As $X$ was split between at least two connected components of $G \triangle F$,
$F$ contains at least one edge of $G[X]$, whereas $F'$ does not.
We infer that $|F'| < |F|$ and the lemma is proven.
\end{proof}

\newcommand{\nice}{\mathcal{N}}

\section{A subexponential algorithm for \pclustering{}}\label{sec:subept-new}

In this section we prove Theorem \ref{thm:pclustering-subept}, that is,
   we show a $\cO(2^{\cO(\sqrt{pk})} + n + m)$-time algorithm for \pclustering.

\subsection{Reduction for large $p$}\label{sec:large-p}

The first step of our algorithm is an application of the kernelization algorithm by Guo~\cite{Guo09} (Proposition~\ref{prop:polykernel_Guo}) followed by some additional preprocessing rules that ensure that $p\leq 6k$. These additional rules are encapsulated in the following lemma; the rest of this section is devoted to its proof.

\begin{lemma}\label{lem:preprocessing}
There exists a polynomial time algorithm that, given an instance $(G,p,k)$ of \pclustering, outputs an equivalent instance $(G',p',k)$, where $G'$ is an induced subgraph of $G$ and $p'\leq 6k$.
\end{lemma}

Before we proceed to formal argumentation, let us provide some intuition. The key idea behind Lemma~\ref{lem:preprocessing} is the observation that if $p>2k$, then at least $p-2k$ clusters in the final cluster graph cannot be touched by the solution, hence they must have been present as isolated cliques already in the beginning. Hence, if $p>6k$ then we have to already see $p-2k>4k$ isolated cliques; otherwise, we may safely provide a negative answer. Although these cliques may be still merged (to decrease the number of clusters) or split (to increase the number of clusters), we can apply greedy arguments to identify a clique that may be safely assumed to be untouched by the solution. Hence we can remove it from the graph and decrement $p$ by one. Although the greedy arguments seem very intuitive, their formal proofs turn out to be somewhat technical.

We now proceed to a formal proof of Lemma~\ref{lem:preprocessing}. Let us fix some optimal solution $F$, i.e., a subset of $V\times V$ of minimum cardinality such that $G\triangle F$ is a $p$-cluster graph.

Consider the case when $p > 6k$. Observe that only $2k$ out of $p$ resulting clusters in $G\triangle F$ can be adjacent to any pair from the set $F$. Hence at least $p-2k$ clusters must be already present in the graph $G$ as connected components being cliques. Therefore, if $G$ contains less than $p-2k$ connected components that are cliques, then $(G,p,k)$ is a NO-instance.

\begin{description}
\item[Rule 1] If $G$ contains less than $p-2k$ connected components that are cliques, answer NO.
\end{description}

As $p>6k$, if Rule 1 was not triggered then we have more than $4k$ connected components that are cliques. The aim is now to apply greedy arguments to identify a component that can be safely assumed to be untouched. As a first step, consider a situation when $G$ contains more than $2k$ isolated vertices. Then at least one of these vertices is not incident to an element of $F$, thus we may delete one isolated
vertex and decrease $p$ by one.

\begin{description}
\item[Rule 2] If $G$ contains $2k+1$ isolated vertices, pick one of them, say $v$, and delete it from $G$. The new instance is $(G\setminus v, p-1, k)$.
\end{description}

We are left with the case where $G$ contains more than $2k$ connected components that are cliques, but not isolated vertices.
At least one of these cliques is untouched by $F$. Note that even though the number of cliques is large, some of them may be merged with other clusters (to
decrease the number of connected components), or split into more clusters
(to increase the number of connected components), and we have no a priori knowledge about which clique will be left untouched. We argue that in both cases, we can greedily
merge or split the smallest possible clusters. Thus, without loss of generality, we can
assume that the largest connected component of $G$ that is a clique is left untouched
in $G \triangle F$. We reduce the input instance $(G,p,k)$ by deleting this cluster
and decreasing $p$ by one.

\begin{description}
\item[Rule 3] If $G$ contains $2k+1$ isolated cliques that are not isolated vertices, pick  a clique $C$ of largest size and delete it from $G$. The new instance is $(G\setminus C, p-1, k)$.
\end{description}

We formally verify safeness of the Rule $3$ by proving the following lemma.
Without loss of generality, we may assume that the solution $F$, among all solutions of minimum cardinality, has minimum possible number of editions
incident to the connected components of $G$ that are cliques of largest size.

\begin{lemma}
Let $D_1,D_2,\ldots,D_\ell$ be connected components of $G$ that are cliques, but not isolated vertices. Assume that $\ell\geq 2k+1$. Then there exists a component $D_i$ that has the largest size among $D_1,D_2,\ldots,D_\ell$ and none of the pairs from $F$ is incident to any vertex of $D_i$.
\end{lemma}
\begin{proof}
Let $C_1,C_2,\ldots,C_p$ be clusters of $G\triangle F$. We say that cluster $C_i$ {\emph{contains}} component $D_j$ if $V(D_j)\subseteq V(C_i)$, and component $D_j$ {\emph{contains}} cluster $C_i$ if $V(C_i)\subseteq V(D_j)$. Moreover, we say that these containments are {\emph{strict}} if $V(D_j)\subsetneq V(C_i)$ or $V(C_i)\subsetneq V(D_j)$, respectively.

\vskip 0.2cm
{\bf{Claim 1.}} {\emph{For every cluster $C_i$ and component $D_j$, either $V(C_i)\cap V(D_j)=\emptyset$, $C_i$ contains $D_j$ or $D_j$ contains $C_i$.}}
\vskip 0.2cm

In order to prove Claim 1 we need to argue that the situation when sets $V(C_i)\cap V(D_j)$, $V(C_i)\setminus V(D_j)$, $V(D_j)\setminus V(C_i)$ are simultaneously nonempty is impossible. Assume otherwise, and without loss of generality assume further that $|V(C_i)\setminus V(D_j)|$ is largest possible. As $V(D_j)\setminus V(C_i)\neq \emptyset$, take some $C_{i'}\neq C_i$ such that $V(C_{i'})\cap V(D_j)\neq \emptyset$. By the choice of $C_i$ we have that $|V(C_i)\setminus V(D_j)|\geq |V(C_{i'})\setminus V(D_j)|$ (note that $V(C_{i'})\setminus V(D_j)$ is possibly empty). Consider a new cluster graph $H$ obtained from $G\triangle F$ by moving $V(C_i)\cap V(D_j)$ from the cluster $C_i$ to the cluster $C_{i'}$. Clearly, $H$ still has $p$ clusters as $V(C_i)\setminus V(D_j)$ is nonempty. Moreover, the edition set $F'$ that needs to be modified in order to obtain $H$ from $G$, differs from $F$ as follows: it additionally contains $(V(C_i)\cap V(D_j))\times (V(C_{i'})\setminus V(D_j))$, but does not contain $(V(C_i)\cap V(D_j))\times (V(C_i)\setminus V(D_j))$ nor $(V(C_i)\cap V(D_j))\times (V(C_{i'})\cap V(D_j))$. As $|V(C_i)\setminus V(D_j)|\geq |V(C_{i'})\setminus V(D_j)|$, we have that
$$|(V(C_i)\cap V(D_j))\times (V(C_{i'})\setminus V(D_j))|\leq |(V(C_i)\cap V(D_j))\times (V(C_i)\setminus V(D_j))|$$
and
$$|(V(C_i)\cap V(D_j))\times (V(C_{i'})\cap V(D_j))|>0.$$
Hence $|F'|<|F|$, which is a contradiction with minimality of $F$. This settles Claim 1.

We say that a component $D_j$ is {\emph{embedded}} if some cluster $C_i$ strictly contains it. Moreover, we say that a component $D_j$ is {\emph{broken}} if it strictly contains more than one cluster; Claim 1 implies that then $V(D_j)$ is the union of vertex sets of the clusters it strictly contains. Component $D_j$ is said to be {\emph{untouched}} if none of the pairs from $F$ is incident to a vertex from $D_j$. Claim 1 proves that every cluster is either embedded, broken or untouched.

\vskip 0.2cm
{\bf{Claim 2.}} {\emph{It is impossible that some component $D_j$ is broken and some other $D_{j'}$ is embedded.}}
\vskip 0.2cm

In order to prove Claim 2 assume, otherwise, that some component $D_j$ is broken and some other $D_{j'}$ is embedded. Let $C_{i_1},C_{i_2}$ be any two clusters contained in $D_j$ and let $C_{i'}$ be the cluster that strictly contains $D_{j'}$. Consider a new cluster graph $H$ obtained from $G\triangle F$ by merging clusters $C_{i_1}$,$C_{i_2}$ and splitting cluster $C_{i'}$ into clusters on vertex sets $V(C_{i'})\setminus V(D_{j'})$ and $V(D_{j'})$. As $V(C_{i'})\setminus V(D_{j'})\neq \emptyset$, $H$ is still a $p$-cluster graph. Moreover, the edition set $F'$ that need to be modified in order to obtain $H$ from $G$, differs from $F$ by not containing $V(C_{i_1})\times V(C_{i_2})$ and $(V(C_{i'})\setminus V(D_{j'}))\times V(D_{j'})$. Both of this sets are nonempty, so $|F'|<|F|$, which is a contradiction with minimality of $F$. This settles Claim 2.

Claim 2 implies that either none of the components is broken, or none is embedded. We firstly prove that in the first case the lemma holds. Note that as $\ell>2k$, at least one component is untouched.

\vskip 0.2cm
{\bf{Claim 3.}} {\emph{If none of the components $D_1,D_2,\ldots,D_\ell$ is broken, then there is an untouched component $D_j$ with the largest number of vertices among $D_1,D_2,\ldots,D_\ell$.}}
\vskip 0.2cm

Assume, otherwise, that all the components with largest numbers of vertices are not untouched, hence they are embedded. Take any such component $D_j$ and let $D_{j'}$ be any untouched component; by the assumption we infer that $|V(D_j)|>|V(D_{j'})|$. Let $C_i$ be the cluster that strictly contains $D_j$ and let $C_{i'}$ be the cluster corresponding to the (untouched) component $D_{j'}$. Consider a cluster graph $H$ obtained from $G\triangle F$ by exchanging sets $V(D_j)$ and $V(D_{j'})$ between clusters $C_i$ and $C_{i'}$. Observe that the edition set $F'$ that needs to be modified in order to obtain $H$ from $G$, differs from $F$ by not containing $(V(C_i)\setminus V(D_j))\times V(D_j)$ but containing $(V(C_i)\setminus V(D_j))\times V(D_{j'})$. However, $|V(D_j)|>|V(D_{j'})|$ and $|V(C_i)\setminus V(D_j)|>0$, so $|F'|<|F|$. This contradicts minimality of $F$ and settles Claim 3.

We are left with the case when all the clusters are broken or untouched.

\vskip 0.2cm
{\bf{Claim 4.}} {\emph{If none of the components $D_1,D_2,\ldots,D_\ell$ is embedded, then there is an untouched component $D_j$ with the largest number of vertices among $D_1,D_2,\ldots,D_\ell$.}}
\vskip 0.2cm

Assume, otherwise, that all the components with largest numbers of vertices are not untouched, hence they are broken. Take any such component $D_j$ and let $D_{j'}$ be any untouched component; by the assumption we infer that $|V(D_j)|>|V(D_{j'})|$. Assume that $D_j$ is broken into $q+1$ clusters ($q\geq 1$) of sizes $a_1,a_2,\ldots,a_{q+1}$, where $\sum_{i=1}^{q+1} a_i=|V(D_j)|$. The number of editions needed inside component $D_j$ is hence equal to 
$$\binom{|V(D_j)|}{2}-\sum_{i=1}^{q+1} \binom{a_i}{2}\geq \binom{|V(D_j)|}{2}-\binom{|V(D_j)|-q}{2}-q\binom{1}{2}=\binom{|V(D_j)|}{2}-\binom{|V(D_j)|-q}{2}.$$
The inequality follows from convexity of the function $t\to \binom{t}{2}$. We now consider two cases. 

Assume first that $|V(D_{j'})|> q$. Let us change the edition set $F$ into $F'$ by altering editions inside components $D_j$ and $D_{j'}$ as follows: instead of breaking $D_j$ into $q+1$ components and leaving $D_{j'}$ untouched, leave $D_j$ untouched and break $D_{j'}$ into $q+1$ components by creating $q$ singleton clusters and one cluster of size $|V(D_{j'})|-q$. Similar calculations to the ones presented in the paragraph above show that the edition cost inside components $D_j$ and $D_{j'}$ is equal to $\binom{|V(D_{j'})|}{2}-\binom{|V(D_{j'})|-q}{2}<\binom{|V(D_{j})|}{2}-\binom{|V(D_{j})|-q}{2}$. Hence, we can obtain the same number of clusters with a strictly smaller edition set, a contradiction with minimality of $F$.

Assume now that $|V(D_{j'})|\leq q$. Let us change the edition set $F$ into $F'$ by altering editions inside components $D_j$ and $D_{j'}$ as follows: instead of breaking $D_j$ into $q+1$ components and leaving $D_{j'}$ untouched, we break $D_{j'}$ totally into $|V(D_{j'})|$ singleton clusters and break $D_j$ into $q-|V(D_{j'})|+1$ singleton clusters and one of size $|V(D_j)|-q+|V(D_{j'})|-1$. Clearly, we get the same number of clusters in this manner. Similar calculations as before show that the number of new editions needed inside clusters $D_j$ and $D_{j'}$ is equal to $\binom{|V(D_{j'})|}{2}+\binom{|V(D_{j})|}{2}-\binom{|V(D_j)|-q+|V(D_{j'})|-1}{2}$, which is not larger than $\binom{|V(D_{j})|}{2}-\binom{|V(D_{j})|-q}{2}$ for $|V(D_{j'})|\geq 2$ (recall that components $D_1,D_2,\ldots,D_\ell$ are not independent vertices).
Hence, we can obtain the same number of clusters with a not larger edition set
and with a smaller number of editions incident to components of $G$ that are cliques of largest size.
This contradicts the choice of $F$.

We have obtained a contradiction in both cases, so Claim 4 follows. Claims 3 and 4 imply the thesis of the lemma.
\end{proof}

Clearly, an instance on which none of the Rules 1--3 may be triggered, has $p\leq 6k$. This proves Lemma~\ref{lem:preprocessing}.

\subsection{Bounds on binomial coefficients}

In the running time analysis we need some combinatorial bounds on binomial coefficients. More precisely, we use the following inequality.

\begin{lemma}\label{lem:sqrt-bound}
If $a,b$ are nonnegative integers, then $\binom{a+b}{a}\leq 2^{2\sqrt{ab}}$.
\end{lemma}

We start with the following simple observation.

\begin{lemma}\label{lem:ind}
If $a,b$ are positive integers, then $\binom{a+b}{a}\leq \frac{(a+b)^{a+b}}{a^ab^b}$.
\end{lemma}
\begin{proof}
In the proof we use a folklore fact that the sequence $a_n=(1+1/n)^n$ is increasing. This implies that $\left(1+\frac{1}{b}\right)^b\leq \left(1+\frac{1}{a+b}\right)^{a+b}$, equivalently $\frac{(a+b)^{a+b}}{b^b}\leq \frac{(a+b+1)^{a+b}}{(b+1)^b}$.

Let us fix $a$; we prove the claim via induction with respect to $b$. For $b=1$ the claim is equivalent to $a^a \leq (a+1)^a$ and therefore trivial.
In order to check the induction step, notice that
\begin{eqnarray*}
\binom{a+b+1}{a} & = & \frac{a+b+1}{b+1}\cdot\binom{a+b}{a}\leq \frac{a+b+1}{b+1}\cdot\frac{(a+b)^{a+b}}{a^ab^b} \\
& \leq & \frac{a+b+1}{b+1}\cdot\frac{(a+b+1)^{a+b}}{a^a(b+1)^b} = \frac{(a+b+1)^{a+b+1}}{a^a(b+1)^{b+1}}.
\end{eqnarray*}
\end{proof}

We proceed to the proof of Lemma~\ref{lem:sqrt-bound}.

\begin{proof}[Proof of Lemma~\ref{lem:sqrt-bound}]
Firstly, observe that the claim is trivial for $a=0$ or $b=0$; hence, we can assume that $a,b>0$. Moreover, without losing generality assume that $a\leq b$. Let us denote $\sqrt{ab}=\ell$ and $\frac{a}{b}=t$, then $0<t\leq 1$. By Lemma~\ref{lem:ind} we have that 
\begin{eqnarray*}
\binom{a+b}{a} & \leq & \frac{(a+b)^{a+b}}{a^ab^b}=\left(1+\frac{b}{a}\right)^a\cdot\left(1+\frac{a}{b}\right)^b \\
& = & \left[\left(1+\frac{1}{t}\right)^{\frac{\sqrt{t}}{1}}\cdot \left(1+\frac{t}{1}\right)^{\frac{1}{\sqrt{t}}}\right]^\ell.
\end{eqnarray*}
Let us denote $g(x)=x\ln\left(1+x^{-2}\right)+x^{-1}\ln\left(1+x^2\right)$. As $\binom{a+b}{a}\leq e^{\ell\cdot g\left(\sqrt{t}\right)}$, it suffices to prove that $g(x)\leq 2\ln 2$ for all $0<x\leq 1$. Observe that
\begin{eqnarray*}
g'(x) & = & \ln\left(1+x^{-2}\right)-x\cdot 2x^{-3}\cdot \frac{1}{1+x^{-2}}-x^{-2}\ln\left(1+x^2\right)+x^{-1}\cdot 2x\cdot \frac{1}{1+x^2} \\
& = & \ln\left(1+x^{-2}\right)-\frac{2}{1+x^2}-x^{-2}\ln\left(1+x^2\right)+\frac{2}{1+x^2} \\
& = & \ln\left(1+x^{-2}\right)-x^{-2}\ln\left(1+x^2\right).
\end{eqnarray*}
Let us now introduce $h:(0,1] \to \mathbb{R}$, $h(y)=g'(\sqrt{y})=\ln\left(1+y^{-1}\right)-y^{-1}\ln\left(1+y\right)$. Then,
\begin{eqnarray*}
h'(y) & = & -y^{-2}\cdot \frac{1}{1+y^{-1}}+y^{-2}\ln\left(1+y\right)-y^{-1}\cdot\frac{1}{1+y} \\
& = & y^{-2}\ln\left(1+y\right)-\frac{2}{y+y^2}.
\end{eqnarray*}
We claim that $h'(y)\leq 0$ for all $y\leq 1$. Indeed, from the inequality $\ln(1+y)\leq y$ we infer that
\begin{eqnarray*}
y^{-2}\ln\left(1+y\right) & \leq & y^{-1} = \frac{2}{y+y} \leq \frac{2}{y+y^2}.
\end{eqnarray*}
Therefore, $h'(y)\leq 0$ for $y\in (0,1]$, so $h(y)$ is non-increasing on this interval. As $h(1)=0$, this implies that $h(y)\geq 0$ for $y\in (0,1]$, so also $g'(x)\geq 0$ for $x\in (0,1]$. This means that $g(x)$ is non-decreasing on the interval $(0,1]$, so $g(x)\leq g(1)=2\ln 2$.
\end{proof}

\subsection{Small cuts}

We now proceed to the algorithm itself. Let us introduce the key notion.

\begin{definition}
Let $G=(V,E)$ be an undirected graph. A partition $(V_1,V_2)$ of $V$ is called a {\emph{$k$-cut of $G$}} if $|E(V_1,V_2)|\leq k$.
\end{definition}

\begin{lemma}\label{lem:enumeration}
$k$-cuts of a graph $G$ can be enumerated with polynomial time delay.
\end{lemma}
\begin{proof}
We follow the standard branching. We order the vertices arbitrarily, start with empty $V_1$, $V_2$ and for each consecutive vertex $v$ we branch into two subcases: we put $v$ either into $V_1$ or into $V_2$. Once the alignment of all vertices is decided, we output the cut. However, each time we put a vertex in one of the sets, we run a polynomial-time max-flow algorithm to check whether the minimum edge cut between $V_1$ and $V_2$ constructed so far is at most $k$. If not, then we terminate this branch as it certainly cannot result in any solutions found. Thus, we always pursue a branch that results in at least one feasible solution, and finding the next solution occurs within a polynomial number of steps.
\end{proof}

Intuitively, $k$-cuts of the graph $G$ form the search space of the algorithm. Therefore, we would like to bound their number. We do this by firstly bounding the number of cuts of a cluster graph, and then using the fact that a YES-instance is not very far from some cluster graph.
We begin with the following bound on binomial coefficients.

To prove Lemma \ref{lem:sqrt-bound} we need the following fact.

\begin{lemma}\label{lem:cluster-bound}
Let $K$ be a cluster graph containing at most $p$ clusters, where $p\leq 6k$. Then the number of $k$-cuts of $K$ is at most $2^{8\sqrt{pk}}$.
\end{lemma}
\begin{proof}
By slightly abusing the notation, assume that $K$ has exactly $p$ clusters, some of which may be empty. Let $C_1,C_2,\ldots,C_p$ be these clusters and $c_1,c_2,\ldots,c_p$ be their sizes, respectively. We firstly establish a bound on the number of cuts $(V_1,V_2)$ such that the cluster $C_i$ contains $x_i$ vertices from $V_1$ and $y_i$ from $V_2$. Then we discuss how to bound the number of ways of selecting pairs $x_i,y_i$ summing up to $c_i$ for which the number of $k$-cuts is positive. Multiplying the obtained two bounds gives us the claim.

Having fixed the numbers $x_i,y_i$, the number of ways in which the cluster $C_i$ can be partitioned is equal to $\binom{x_i+y_i}{x_i}$. Note that $\binom{x_i+y_i}{x_i}\leq 2^{2\sqrt{x_iy_i}}$ by Lemma~\ref{lem:sqrt-bound}. Observe that there are $x_iy_i$ edges between $V_1$ and $V_2$ inside the cluster $C_i$, so if $(V_1,V_2)$ is a $k$-cut, then $\sum_{i=1}^p x_iy_i\leq k$. By applying the Cauchy-Schwarz inequality we infer that $\sum_{i=1}^{p}\sqrt{x_iy_i}\leq \sqrt{p}\cdot \sqrt{\sum_{i=1}^p x_iy_i}\leq \sqrt{pk}$. Therefore, the number of considered cuts is bounded by
$$\prod_{i=1}^p \binom{x_i+y_i}{x_i}\leq 2^{2\sum_{i=1}^p \sqrt{x_iy_i}}\leq 2^{2\sqrt{pk}}.$$
Moreover, observe that $\min(x_i,y_i)\leq \sqrt{x_iy_i}$; hence, $\sum_{i=1}^p \min(x_i,y_i)\leq \sqrt{pk}$. Thus, the choice of $x_i,y_i$ can be modeled by first choosing for each $i$, whether $\min(x_i,y_i)$ is equal to $x_i$ or to $y_i$, and then expressing $\lfloor \sqrt{pk} \rfloor$ as the sum of $p+1$ nonnegative numbers: $\min(x_i,y_i)$ for $1\leq i\leq p$ and the rest, $\lfloor \sqrt{pk} \rfloor-\sum_{i=1}^p \min(x_i,y_i)$. The number of choices in the first step is equal to $2^p\leq 2^{\sqrt{6pk}}$, and in the second is equal to $\binom{\lfloor \sqrt{pk} \rfloor+p}{p}\leq 2^{\sqrt{pk}+\sqrt{6pk}}$. Therefore, the number of possible choices of $x_i,y_i$ is bounded by $2^{(1+2\sqrt{6})\sqrt{pk}}\leq 2^{6\sqrt{pk}}$. Hence, the total number of $k$-cuts is bounded by $2^{6\sqrt{pk}}\cdot 2^{2\sqrt{pk}}=2^{8\sqrt{pk}}$, as claimed.
\end{proof}

\begin{lemma}\label{lem:yes-bound}
If $(G,p,k)$ is a YES-instance of \pclustering{} with $p\leq 6k$, then the number of $k$-cuts of $G$ is bounded by $2^{8\sqrt{2pk}}$.
\end{lemma}
\begin{proof}
Let $K$ be a cluster graph with at most $p$ clusters such that $\ham(G,K)\leq k$. Observe that every $k$-cut of $G$ is also a $2k$-cut of $K$, as $K$ differs from $G$ by at most $k$ edge modifications. The claim follows from Lemma~\ref{lem:cluster-bound}.
\end{proof}

\subsection{The algorithm}

\begin{proof}[Proof of Theorem \ref{thm:pclustering-subept}]
Let $(G=(V,E),p,k)$ be the given \pclustering{} instance. By making use of Proposition~\ref{prop:polykernel_Guo}, we can assume that $G$ has at most $(p+2)k+p$ vertices, thus all the factors polynomial in the size of $G$ can be henceforth hidden within the $2^{\cO(\sqrt{pk})}$ factor. Application of Proposition~\ref{prop:polykernel_Guo} gives the additional $\cO(n+m)$ summand to the complexity. By further usage of Lemma~\ref{lem:preprocessing} we can also assume that $p\leq 6k$. Note that application of Lemma~\ref{lem:preprocessing} can spoil the bound $|V(G)|\leq (p + 2)k + p$ as $p$ can decrease; however the number of vertices of the graph is still bounded in terms of initial $p$ and $k$.

We now enumerate $k$-cuts of $G$ with polynomial time delay. If we exceed the bound $2^{8\sqrt{2pk}}$ given by Lemma~\ref{lem:yes-bound}, we know that we can safely answer NO, so we immediately terminate the computation and give a negative answer. Therefore, we can assume that we have computed the set $\nice$ of all $k$-cuts of $G$ and $|\nice|\leq 2^{8\sqrt{2pk}}$.

Assume that $(G,p,k)$ is a YES-instance and let $K$ be a cluster graph with at most $p$ clusters such that $\ham(G,K)\leq k$. Again, let $C_1,C_2,\ldots,C_p$ be the clusters of $K$. Observe that for every $j\in \{0,1,2,\ldots,p\}$, the partition $\left(\bigcup_{i=1}^j V(C_i),\bigcup_{i=j+1}^p V(C_i)\right)$ has to be the $k$-cut with respect to $G$, as otherwise there would be more than $k$ edges that need to be deleted from $G$ in order to obtain $K$. This observation enables us to use a dynamic programming approach on the set of cuts.

We construct a directed graph $D$, whose vertex set is equal to $\nice\times\{0,1,2,\ldots,p\}\times \{0,1,2,\ldots,k\}$; note that $|V(D)|=2^{\cO(\sqrt{pk})}$. We create arcs going from $((V_1,V_2),j,\ell)$ to $((V_1',V_2'),j+1,\ell')$, where $V_1\subsetneq V_1'$ (hence $V_2\supsetneq V_2'$), $j\in \{0,1,2,\ldots,p-1\}$ and $\ell'=\ell+|E(V_1,V_1'\setminus V_1)|+|\overline{E}(V_1'\setminus V_1,V_1'\setminus V_1)|$ ($(V,\overline{E})$ is the complement of the graph $G$). The arcs can be constructed in $2^{\cO(\sqrt{pk})}$ time by checking for all the pairs of vertices whether they should be connected. We claim that the answer to the instance $(G,p,k)$ is equivalent to reachability of any of the vertices of form $((V,\emptyset),p,\ell)$ from the vertex $((\emptyset,V),0,0)$.

In one direction, if there is a path from $((\emptyset,V),0,0)$ to $((V,\emptyset),p,\ell)$ for some $\ell\leq k$, then the consecutive sets $V_1'\setminus V_1$ along the path form clusters $C_i$ of a cluster graph $K$, whose editing distance to $G$ is accumulated on the last coordinate, thus bounded by $k$. In the second direction, if there is a cluster graph $K$ with clusters $C_1,C_2,\ldots,C_p$ within editing distance at most $k$ from $G$, then vertices $\left(\left(\bigcup_{i=1}^j V(C_i),\bigcup_{i=j+1}^p V(C_i)\right),j,\ham\left(G\left[\bigcup_{i=1}^j V(C_i)\right],K\left[\bigcup_{i=1}^j V(C_i)\right]\right)\right)$ form a path from $((\emptyset,V),0,0)$ to $((V,\emptyset),p,\ham(G,K))$. Note that all these triples are indeed vertices of the graph $D$, as $\left(\bigcup_{i=1}^j V(C_i),\bigcup_{i=j+1}^p V(C_i)\right)$ are $k$-cuts of $G$.

Reachability in a directed graph can be tested in linear time with respect to the number of vertices and arcs. We can now apply this algorithm to the graph $D$ and conclude solving the \pclustering{} instance in $\cO(2^{\cO(\sqrt{pk})} + n + m)$ time.
\end{proof}

\newcommand{\vars}{\mathrm{Vars}}
\newcommand{\ipart}{r}
\newcommand{\isix}{\alpha}
\newcommand{\icyc}{\beta}
\newcommand{\ivcyc}{\gamma}
\newcommand{\icop}{\xi}
\newcommand{\ivar}{\eta}
\newcommand{\lsgn}{\mathrm{sgn}}
\newcommand{\setQ}{\mathcal{Q}}
\newcommand{\setX}{\mathcal{W}}
\newcommand{\setS}{\mathcal{S}}
\newcommand{\kQQ}{k_{\setQ-\setQ}}
\newcommand{\kQXS}{k_{\setQ-\setX\setS}}
\newcommand{\kXSXS}{k_{\setX\setS-\setX\setS}^\mathrm{all}}
\newcommand{\kcut}{k_{\setX\setS-\setX\setS}^\mathrm{exist}}
\newcommand{\kXX}{k_{\setX-\setX}^\mathrm{save}}
\newcommand{\kXS}{k_{\setX-\setS}^\mathrm{save}}
\newcommand{\lQQ}{\ell_{\setQ-\setQ}}
\newcommand{\lQXS}{\ell_{\setQ-\setX\setS}}
\newcommand{\lXSXS}{\ell_{\setX\setS-\setX\setS}^\mathrm{all}}
\newcommand{\lXX}{\ell_{\setX-\setX}^\mathrm{save}}
\newcommand{\lXS}{\ell_{\setX-\setS}^\mathrm{save}}
\newcommand{\XSalone}{n^\mathrm{alone}}
\newcommand{\opvar}{\mathrm{var}}

\section{Multivariate lower bound: proof of Theorem \ref{thm:multivariate-reduction}}\label{sec:multi}

This section contains the proof of Theorem \ref{thm:multivariate-reduction}. The proof consists of four parts. In Section \ref{app:multi-preprocess} we preprocess the input formula $\Phi$ to make it more regular.
Section \ref{app:multi-construction} contains the details of the construction of the graph $G$.
In Section \ref{app:multi-1} we show how to translate a satisfying assignment of $\Phi$ into a $6p$-cluster graph
$G_0$ close to $G$ and we provide a reverse implication in Section \ref{app:multi-2}.
In the proof we treat $\varepsilon$ as a constant and hide the factors depending on it in the $\cO$-notation.
That is, the constants in the $\cO$-notation correspond to the factor $\delta$ in the statement of Theorem \ref{thm:multivariate-reduction}.

\subsection{Preprocessing of the formula}\label{app:multi-preprocess}

We start with a step that regularizes the input formula $\Phi$, while increasing its size only by a constant factor.
The purpose of this step is to ensure that, when we translate a satisfying assignment of $\Phi$ into a cluster graph $G_0$ in
the completeness step, the clusters are of the same size, and therefore contain the minimum possible number of edges.
This property is used in the argumentation of the soundness step.

\begin{lemma}\label{lem:preparation}
For any fixed $\varepsilon > 0$, there exists a polynomial-time algorithm that, given a $3$-CNF formula $\Phi$ with
$n$ variables and $m$ clauses and an integer $p$, $\varepsilon p \leq n$, constructs a $3$-CNF formula $\Phi'$ with $n'$ variables and $m'$ clauses together with
a partition of the variable set $\vars(\Phi')$ into $p$ parts $\vars^\ipart$, $1 \leq \ipart \leq p$, such that following properties hold:
\begin{itemize}
\item[(a)] $\Phi'$ is satisfiable iff $\Phi$ is;
\item[(b)] in $\Phi'$ every clause contains exactly three literals corresponding to different variables;
\item[(c)] in $\Phi'$ every variable appears exactly three times positively and exactly three times negatively;
\item[(d)] $n'$ is divisible by $p$ and, for each $1 \leq \ipart \leq p$, $|\vars^\ipart| = n'/p$
(i.e., the variables are split evenly between the parts $\vars^\ipart$);
\item[(e)] if $\Phi'$ is satisfiable, then there exists a satisfying assignment of $\vars(\Phi')$ with the property
that in each part $\vars^\ipart$ the numbers of variables set to true and to false are equal.
\item[(f)] $n'+m'=\cO(n+m)$, where the constant hidden in the $\cO$-notation depends on $\varepsilon$.
\end{itemize}
\end{lemma}
\begin{proof}
We modify $\Phi$ while preserving satisfiability, consecutively ensuring that properties (b), (c), (d), and (e) are satisfied. Satisfaction of (f) will follow directly from the constructions used.

First, delete every clause that contains two different literals corresponding to the same variable, as they are always satisfied. Remove copies of the same literals inside clauses. Until all the clauses have at least two literals, remove every clause containing one literal, set the value of this literal so that the clause is satisfied and propagate this knowledge to the other clauses. At the end, create a new variable $p$ and for every clause $C$ that has two literals replace it with two clauses $C\vee p$ and $C\vee \neg p$. All these operations preserve satisfiability and at the end all the clauses consist of exactly three different literals.

Second, duplicate each clause so that every variable appears an even number of times. Introduce two new variables $q,r$. Take any variable $x$, assume that $x$ appears positively $k^+$ times and negatively $k^-$ times. If $k^+<k^-$, introduce clauses $(x\vee q\vee r)$ and $(x\vee \neg q\vee \neg r)$, each $\frac{k^- - k^+}{2}$ times, otherwise introduce clauses $(\neg x\vee q\vee r)$ and $(\neg x\vee \neg q\vee \neg r)$, each $\frac{k^+ - k^-}{2}$ times. These operations preserve satisfiability (as the new clauses can be satisfied by setting $q$ to true and $r$ to false) and, after the operation, every variable appears the same number of time positively as negatively (including the new variables $q,r$).

Third, copy each clause three times. For each variable $x$, replace all occurrences of the variable $x$ with a cycle of implications in the following way. Assume 
that $x$ appears $6d$ times (the number of appearances is divisible by six due to the modifications in the previous paragraph and the copying step). Introduce new variables
$x_i$ for $1 \leq i \leq 3d$, $y_i$ for $1 \leq i \leq d$ and clauses $(\neg x_i \vee x_{i+1} \vee y_{\lceil i/3 \rceil})$ and $(\neg x_i \vee x_{i+1} \vee \neg y_{\lceil i/3 \rceil})$
for $1 \leq i \leq 3d$ (with $x_{3d+1} = x_1$).
Moreover, replace each occurrence of the variable $x$ with one of the variables $x_i$ in such a way that each variable $x_i$ is used once in a positive literal and once in a negative one.
In this manner each variable $x_i$ and $y_i$ is used exactly three times in a positive literal and three times in a negative one. Moreover, the new clauses form an implication cycle
$x_1 \Rightarrow x_2 \Rightarrow \ldots \Rightarrow x_{3d} \Rightarrow x_1$, ensuring that all the variables $x_i$ have equal value in any satisfying assignment of the formula.

Fourth, to make $n'$ divisible by $p$ we first copy the entire formula three times, creating a new set of variables for each copy. In this way we ensure that the number
of variables is divisible by three. Then we add new variables in triples to make the number of variables divisible by $p$.
For each triple $x,y,z$ of new variables, we introduce six new clauses: all possible clauses that contain one literal for each variable $x$, $y$ and $z$ except for
$(x \vee y \vee z)$ and $(\neg x \vee \neg y \vee \neg z)$. Note that the new clauses are easily satisfied by setting all new variables to true, while all new variables
appear exactly three times positively and three times negatively. Moreover, as initially $\varepsilon p \leq n$, this step increases the size of the formula only by a constant factor.

Finally, to achieve (d) and (e) take $\Phi'=\Phi\wedge\overline{\Phi}$, where $\overline{\Phi}$ is a copy of $\Phi$ on a disjoint copy of the variable set and with all literals reversed, i.e., positive occurrences are replaced by negative ones and vice versa. Of course, if $\Phi'$ is satisfiable then $\Phi$ as well, while if $\Phi$ is satisfiable, then we can copy the assignment to the copies of variables and reverse it, thus obtaining a feasible assignment for $\Phi'$.
Recall that before this step the number of variables was divisible by $p$. We can now partition the variable set into $p$ parts, such that whenever we include a variable into one part,
we include its copy in the same part as well. 
In order to prove that the property (e) holds, take any feasible solution to $\Phi'$, truncate the evaluation to $\vars(\Phi)$ and copy it while reversing on $\overline{\Phi}$.
\end{proof}

\subsection{Construction}\label{app:multi-construction}

In this section we show how to compute the graph $G$ and the integer $k'$
from the formula $\Phi'$ given by Lemma \ref{lem:preparation}.
As Lemma \ref{lem:preparation} increases the size of the formula
by a constant factor, we have that $n',m' = \cO(\sqrt{pk})$ and $|\vars^\ipart| = n'/p = \cO(\sqrt{k/p})$
for $1 \leq \ipart \leq p$.

Observe that in the statement of the Theorem \ref{thm:multivariate-reduction} we can safely assume that $\varepsilon\leq 1$, as the assumptions become more and more restricted as $\varepsilon$ becomes smaller. From now on we assume that $\varepsilon\leq 1$.

Let $L = 1000\cdot\left(1+ \frac{n'}{p\varepsilon}\right) = \cO(\sqrt{k/p})$.
For each part $\vars^\ipart$, $1 \leq \ipart \leq p$,
we create six cliques $Q_\isix^\ipart$, $1 \leq \isix \leq 6$, each of size $L$.
Let $\setQ$ be the set of all vertices of all cliques $Q_\isix^\ipart$.
In this manner we have $6p$ cliques. Intuitively, if we seek for a $6p$-cluster graph close to $G$, then the cliques are large
enough so that merging two cliques is expensive --- in the intended solution we have exactly one clique in each cluster.

For every variable $x \in \vars^\ipart$, we create six vertices
$w_{1,2}^x, w_{2,3}^x,\ldots,w_{5,6}^x,w_{6,1}^x$. Connect them into a cycle in this order; this cycle is called a {\em{$6$-cycle for the variable $x$}}.
Moreover, for each $1 \leq \isix \leq 6$ and $v \in V(Q_\isix^\ipart)$,
create edges $vw_{\isix-1,\isix}^x$ and $vw_{\isix,\isix+1}^x$
(we assume that the indices behave cyclically,
i.e., $w_{6,7}^x = w_{6,1}^x$, $Q_7^\ipart = Q_1^\ipart$ etc.).
Let $\setX$ be the set of all vertices $w_{\isix,\isix+1}^x$ for all variables $x$.
Intuitively, the cheapest way to cut the $6$-cycle for variable $x$ is to assign the vertices
$w_{\isix,\isix+1}^x$, $1 \leq \isix \leq 6$ all either to the clusters
with cliques with only odd indices or only with even indices. Choosing even indices corresponds to setting $x$ to false, while choosing odd ones corresponds to setting $x$ to true.

Let $\ipart(x)$ be the index of the part that contains the variable $x$,
    that is, $x \in \vars^{\ipart(x)}$.

In each clause $C$ we (arbitrarily) enumerate variables: for $1 \leq \ivar \leq 3$, let $\opvar(C,\ivar)$ be the variable in the $\ivar$-th literal of $C$,
and $\lsgn(C,\ivar) = 0$ if the $\ivar$-th literal is negative and $\lsgn(C,\ivar)=1$ otherwise.

For every clause $C$ create nine vertices: $s_{\icyc,\icop}^C$ for $1 \leq \icyc,\icop \leq 3$.
The edges incident to the vertex $s_{\icyc,\icop}^C$ are defined as follows:
\begin{itemize}
\item for each $1 \leq \ivar \leq 3$ create an edge $s_{\icyc,\icop}^Cw_{2\icyc+2\ivar-3,2\icyc+2\ivar-2}^{\opvar(C,\ivar)}$;
\item if $\icop = 1$, for each $1 \leq \ivar \leq 3$ connect $s_{\icyc,\icop}^C$ to all vertices of one
of the cliques the vertex $w_{2\icyc+2\ivar-3,2\icyc+2\ivar-2}^{\opvar(C,\ivar)}$ is adjacent to depending
on the sign of the $\ivar$-th literal in $C$, that is, the clique $Q_{2\icyc + 2\ivar - 2 - \lsgn(C,\ivar)}^{\ipart(\opvar(C,\ivar))}$;
\item if $\icop > 1$, for each $1 \leq \ivar \leq 3$ connect $s_{\icyc,\icop}^C$ to all vertices of both
cliques the vertex $w_{2\icyc+2\ivar-3,2\icyc+2\ivar-2}^{\opvar(C,\ivar)}$ is adjacent to, that is,
the cliques $Q_{2\icyc + 2\ivar - 3}^{\ipart(\opvar(C,\ivar))}$ and $Q_{2\icyc + 2\ivar - 2}^{\ipart(\opvar(C,\ivar))}$.
\end{itemize}
We note that for a fixed vertex $s_{\icyc,\icop}^C$, the aforementioned cliques $s_{\icyc,\icop}^C$ is adjacent to
are pairwise different, and they have pairwise different subscripts (but may have equal superscripts, i.e., belong to the same part).
See Figure \ref{fig:example-construction} for an illustration.

Let $\setS$ be the set of all vertices $s_{\icyc,\icop}^C$ for all clauses $C$.
If we seek a $6p$-cluster graph close to the graph $G$,
it is reasonable
to put a vertex $s_{\icyc,\icop}^C$ in a cluster together with one of the cliques
this vertex is attached to. If $s_{\icyc,\icop}^C$ is put in a cluster together with
one of the vertices $w_{2\icyc+2\ivar-3,2\icyc+2\ivar-2}^{\opvar(C,\ivar)}$ for $1 \leq \ivar \leq 3$,
we do not need to cut the appropriate edge. The vertices $s_{\icyc,1}^C$ verify
the assignment encoded by the variable vertices $w_{\isix,\isix+1}^x$;
the vertices $s_{\icyc,2}^C$ and $s_{\icyc,3}^C$ help us to make all clusters
be of equal size (which is helpful in the soundness argument).

%\begin{figure}[htb]
%\centering
%\includegraphics[scale=1]{figures/sixcluster-figure.pdf}
%\caption{A part of the graph $G$ created for the variable $x$ and its interaction with vertex
% $s_{1,1}^C$ for $C=x\vee y\vee \neg z$, for a special case
%   $\ipart = \ipart(x) = \ipart(y) = \ipart(z)$.}
%\label{fig:sixcluster}
%\end{figure}

\begin{figure}[htbp]
\centering
\includegraphics[scale=1]{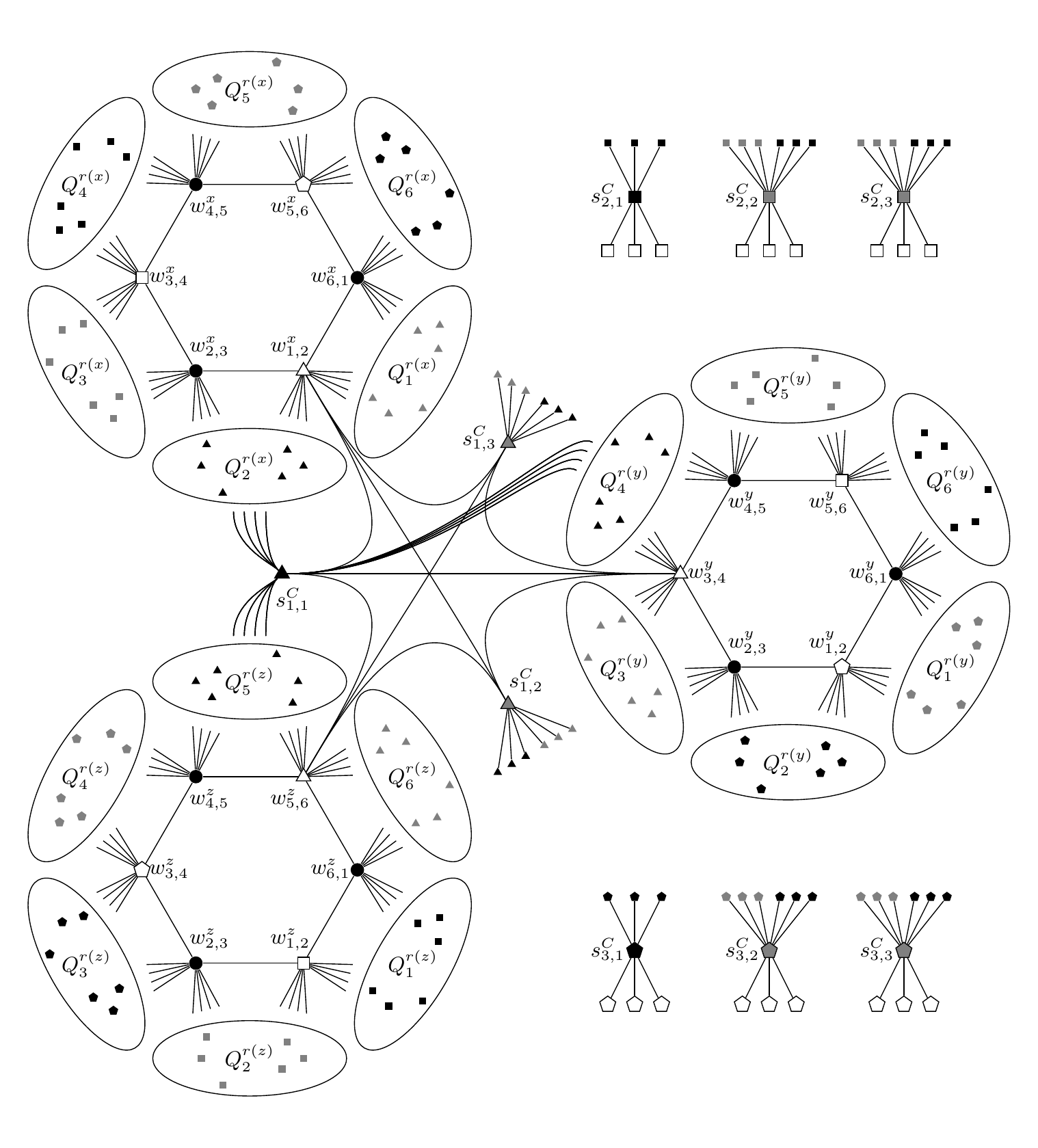}
\caption{A part of the graph $G$ created for the clause $C=(\neg x \vee \neg y \vee z)$, with
  $\opvar(C,1) = x$, $\opvar(C,2) = y$ and $\opvar(C,3) = z$.
    Note that the parts $\ipart(x)$, $\ipart(y)$ and $\ipart(z)$ may be not be pairwise distinct.
    However, due to the rotation index $\icyc$, in any case for a fixed vertex $s_{\icyc,\icop}^C$
    the cliques this vertex is adjacent to on this figure are pairwise distinct and have pairwise distinct subscripts.}
\label{fig:example-construction}
\end{figure}

We note that $|V(G)| = 6pL + \cO(n'+m') = \cO(\sqrt{pk})$.

We now define the budget $k'$ for edge editions. To make the presentation more clear,
we split this budget into few summands.
Let
\begin{align*}
\kQQ &= 0, & \kQXS &= (6n' + 36m')L, & \kXSXS &= 6p \binom{\frac{6n' + 9m'}{6p}}{2}, \\
\kcut &= 6n' + 27m', & \kXX &= 3n', & \kXS &= 9m' \\
 \end{align*}
and finally
$$k' = \kQQ + \kQXS + \kXSXS + \kcut - 2\kXX - 2\kXS.$$
Note that, as $p \leq k$, $L = \cO(\sqrt{k/p})$ and $n',m' = \cO(\sqrt{pk})$, we have
$k' = \cO(k)$.

The intuition behind this split is as follows.
The intended solution for the \pclustering{} instance $(G,6p,k')$
creates no edges between the cliques $Q_\isix^\ipart$, each clique is contained
in its own cluster, and $\kQQ = 0$. For each $v \in \setX \cup \setS$,
the vertex $v$ is assigned to a cluster with one clique $v$ is adjacent to;
$\kQXS$ accumulates the cost of removal of other edges in $E(\setQ,\setX \cup \setS)$.
Finally, we count the editions in $(\setX \cup \setS) \times (\setX \cup \setS)$ in
an indirect way. First we cut all edges of $E(\setX \cup \setS, \setX \cup \setS)$
(summand $\kcut$). We group the vertices of $\setX \cup \setS$ into clusters
and add edges between vertices in each cluster; the summand $\kXSXS$
corresponds to the cost of this operation when all the clusters are of the same size
(and the number of edges is minimum possible). Finally, in summands
$\kXX$ and $\kXS$ we count how many edges are removed and then added again
in this process: $\kXX$ corresponds to saving three edges from each
$6$-cycle in $E(\setX,\setX)$ and $\kXS$ corresponds to saving one
edge in $E(\setX,\setS)$ per each vertex $s_{\icyc,\icop}^C$.

\subsection{Completeness}\label{app:multi-1}

We now show how to translate a satisfying assignment of the input formula
$\Phi$ into a $6p$-cluster graph close to $G$.
\begin{lemma}
If the input formula $\Phi$ is satisfiable, then there exists a $6p$-cluster graph
$G_0$ on vertex set $V(G)$ such that $\ham(G,G_0) = k'$.
\end{lemma}
\begin{proof}
Let $\phi'$ be a satisfying assignment of the formula $\Phi'$ as guaranteed by Lemma \ref{lem:preparation}. Recall that in each
part $\vars^\ipart$, the assignment $\phi'$ sets the same number
of variables to true as to false.
% Let $\phi'$ be a satisfying assignment of the formula $\Phi'$,
% constructed in the preprocessing step of the construction
% (Lemma \ref{lem:preparation}). Recall that in each
% part $\vars^\ipart$, the assignment $\phi'$ sets the same number
% of variables to true as to false.

To simplify the presentation, we identify the range of $\phi'$ with integers: $\phi'(x) = 0$ if $x$ is evaluated to false in $\phi'$ and $\phi'(x) = 1$ otherwise.
Moreover, for a clause $C$ by $\ivar(C)$ we denote the index of an arbitrarily chosen literal that satisfies $C$ in the assignment $\phi'$.

We create $6p$ clusters $K_\isix^\ipart$, $1 \leq \ipart \leq p$, $1 \leq \isix \leq 6$,
as follows:
\begin{itemize}
\item $Q_\isix^\ipart \subseteq K_\isix^\ipart$ for $1 \leq \ipart \leq p$, $1 \leq \isix \leq 6$;
\item for $x\in \vars(\Phi')$, if $\phi'(x)=1$ then $w_{6,1}^x,w_{1,2}^x\in K_1^{\ipart(x)}$, $w_{2,3}^x,w_{3,4}^x\in K_3^{\ipart(x)}$, $w_{4,5}^x,w_{5,6}^x\in K_5^{\ipart(x)}$;
\item for $x\in \vars(\Phi')$, if $\phi'(x)=0$ then $w_{1,2}^x,w_{2,3}^x\in K_2^{\ipart(x)}$, $w_{3,4}^x,w_{4,5}^x\in K_4^{\ipart(x)}$, $w_{5,6}^x,w_{6,1}^x\in K_6^{\ipart(x)}$;
\item for each clause $C$ of $\Phi'$ and $1 \leq \icyc,\icop \leq 3$ we define $\ivar = \ivar(C) + \icop - 1$ and we assign the vertex $s_{\icyc,\icop}^C$ to the cluster
   $K_{2\icyc+2\ivar - 2 - \phi'(\opvar(C,\ivar))}^{\ipart(\opvar(C,\ivar))}$.
\end{itemize}
Note that in this way $s_{\icyc,\icop}^C$ belongs to the same cluster as its neighbor $w_{2\icyc+2\ivar - 3,2\icyc+2\ivar - 2}^{\opvar(C,\ivar)}$.
See Figure \ref{fig:example-assignment} for an illustration.

\begin{figure}[htbp]
\centering
\includegraphics[scale=1]{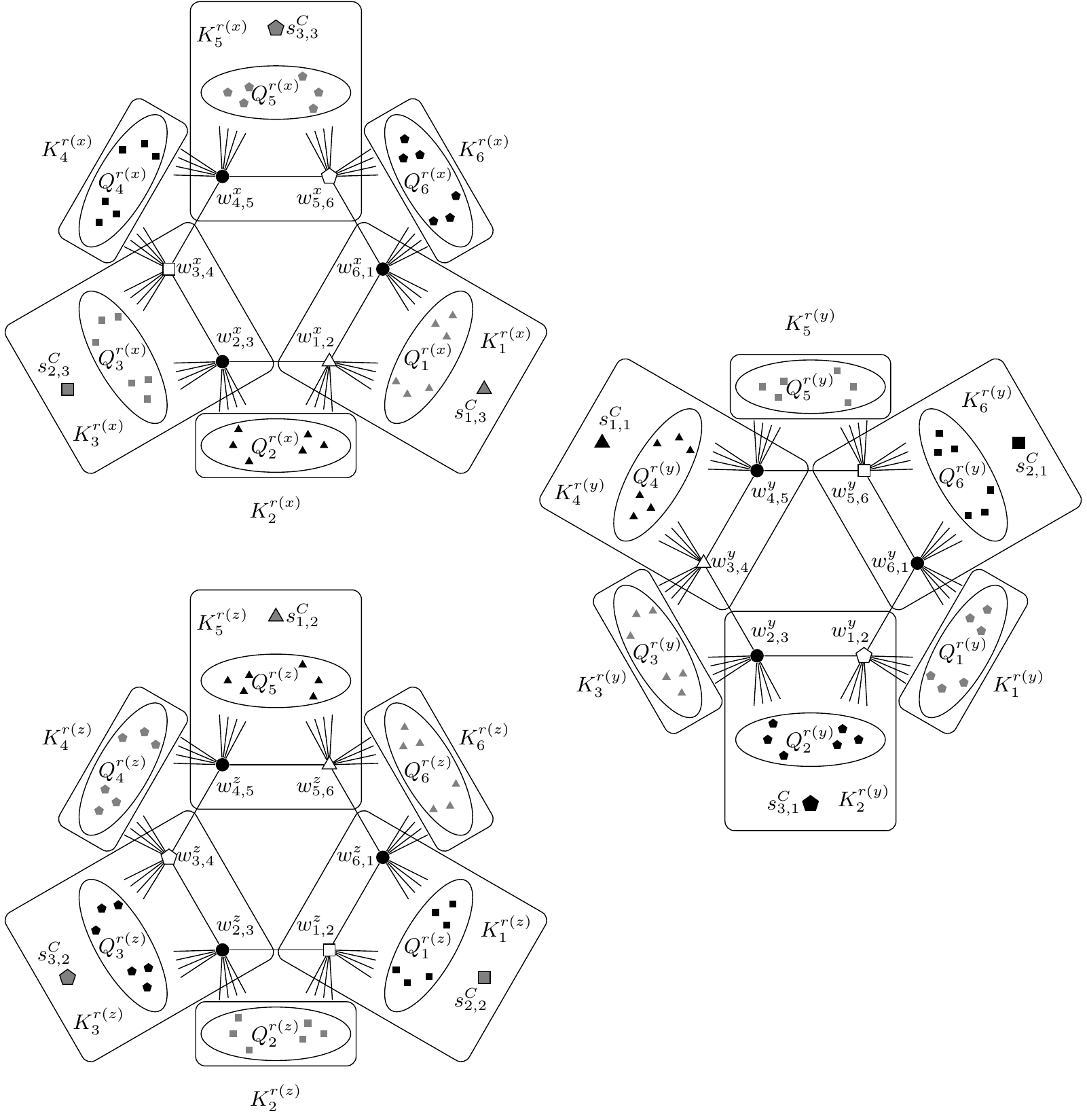}
\caption{Parts of clusters for variables $x$, $y$ and $z$ with $\phi'(x) = 1$, $\phi'(y) = 0$, $\phi'(z) = 1$,
  and a clause $C=(\neg x \vee \neg y \vee z)$ with $\opvar(C,1) = x$, $\opvar(C,2) = y$, $\opvar(C,3) = z$ and $\ivar(C) = 2$
    (note that both $y$ and $z$ satisfy $C$ in the assignment $\phi'$, but $y$ was chosen as a representative).}
\label{fig:example-assignment}
\end{figure}

Let us now compute $\ham(G,G_0)$.
We do not need to add nor delete any edges in $G[\setQ]$.
We note that each vertex $v \in \setX \cup \setS$ is assigned to a cluster with one clique $Q_\isix^\ipart$
it is adjacent to. Indeed, this is only non-trivial for vertices $s_{\icyc,1}^C$ for clauses $C$ and $1 \leq \icyc \leq 3$.
Note that this vertex belongs to the same cluster as the vertex $w_{2\icyc + 2\ivar(C)- 2 - \phi'(\opvar(C,\ivar(C)))}^{\opvar(C,\ivar(C))}$,
and, since the $\ivar(C)$-th literal of $C$ satisfies $C$ in the assignment $\phi'$, $s_{\icyc,1}^C$ is adjacent to all vertices
of the clique $Q_{2\icyc+2\ivar(C) - 2 - \phi'(\opvar(C,\ivar(C)))}^{\ipart(\opvar(C,\ivar(C)))}$.

Therefore we need to cut $\kQXS = (6n'+36m')L$
edges in $E(\setQ,\setX \cup \setS)$:
$L$ edges adjacent to each vertex $w_{\isix,\isix+1}^x$, %$x_{\isix,\isix+1}$,
$2L$ edges adjacent to each vertex $s_{\icyc,1}^C$, and $5L$ edges adjacent to
each vertex $s_{\icyc,2}^C$ and $s_{\icyc,3}^C$. We do not add any new edges between
$\setQ$ and $\setX \cup \setS$.

To count the number of editions in $G[\setX \cup \setS]$, let us first verify that
the clusters $K_\isix^\ipart$ are of equal sizes.
Fix cluster $K_\isix^\ipart$, $1 \leq \isix \leq 6$, $1 \leq \ipart \leq p$.
$K_\isix^\ipart$ contains two vertices $w_{\isix-1,\isix}^x$
and $w_{\isix,\isix+1}^x$ for each variable $x$ with $\phi'(x) + \isix$ being even.
Since $\phi'$ evaluates the same number
of variables in $\vars^\ipart$ to true as to false, we infer that each cluster $K_\isix^\ipart$
contains exactly $n'/p$ vertices from $\setX$, corresponding to $n'/(2p) = |\vars^\ipart|/2$
variables
of $\vars^\ipart$.

For $1 \leq \isix \leq 6$, let $\vars^\ipart_\isix = \phi^{-1}(0) \cap \vars^\ipart$ if $\isix$ is even
and $\phi^{-1}(1) \cap \vars^\ipart$ if $\isix$ is odd. That is, $x \in \vars^\ipart_\isix$ if and only if
$w_{\isix-1,\isix}^x,w_{\isix,\isix+1}^x \in K_\isix^\ipart$.
By the properties of $\Phi'$, for each $x \in \vars^\ipart_\isix$ the variable $x$ appears in three clauses positively
and in three clauses negatively; in particular, it satisfies exactly three clauses in
the assignment $\phi'$. We claim that $K_\isix^\ipart \cap \setS$ consists of $3|\vars^\ipart_\isix| = \frac{3}{2}|\vars^\ipart|$
vertices, that is, for each variable $x \in \vars^\ipart_\isix$, for each clause $C$ (out of three) that $x$ satisfies in the assignment
$\phi'$, $K_\isix^\ipart$ contains exactly one (out of nine) vertex $s_{\icyc,\icop}^C$, and no more vertices of $\setS$.

In one direction, take a variable $x \in \vars^\ipart_\isix$ and a clause $C$ that is satisfied by $x$ in the assignment $\phi'$.
Let $\isix' = 2\lceil \isix/2 \rceil$, so that $w_{\isix'-1,\isix'}^x \in \{w_{\isix-1,\isix}^x,w_{\isix,\isix+1}^x\}$ is the
vertex with first subscript odd and the second even. Take $\ivar$ such that $x = \opvar(C,\ivar)$ and $\icyc = \isix'/2 - \ivar + 1$.
Then $\isix' = 2\icyc + 2\ivar - 2$, and the three vertices $s_{\icyc,\icop}^C$ for $1\leq \icop\leq 3$ are adjacent to $w_{\isix'-1,\isix'}^x$.
Now let $\icop = \ivar - \ivar(C) + 1$; then $s_{\icyc,\icop}^C$ is assigned to the same cluster as $w_{\isix'-1,\isix'}^x$
since $\ivar = \ivar(C) + \icop - 1$. Since $x \in \vars^\ipart_\isix$, then $s_{\icyc,\icop}^C\in K_\isix^\ipart$.

In the other direction, let $s_{\icyc,\icop}^C \in K_\isix^\ipart$ for some clause $C$ and $1 \leq \icyc,\icop \leq 3$.
Recall that $s_{\icyc,\icop}^C$ belongs to the same cluster as one of its three neighbors in $\setX$.
Therefore there exists $w_{\isix'-1,\isix'}^x$ adjacent to $s_{\icyc,\icop}^C$ that belongs to $K_\isix^\ipart$; note that $\isix'$ is even. Moreover, as $s_{\icyc,\icop}^C$ and $w_{\isix'-1,\isix'}^x$ are assigned to the same cluster, we infer that $x$ satisfies $C$.
As $w_{\isix'-1,\isix'}^x \in K_\isix^\ipart$, then $x \in \vars^\ipart_\isix$. Let $\ivar$ be such that $x = \opvar(C,\ivar)$.
As $s_{\icyc,\icop}^Cw_{\isix'-1,\isix'}^x \in E(G)$, we have $\isix'/2 = \icyc + \ivar - 1$, that is, $\icyc = \isix'/2 - \ivar + 1$.
Recall that the neighbors of $s_{\icyc,\icop}^C$ from $\setX$ have pairwise different subscripts; that is, $s_{\icyc,\icop}^C$
is adjacent to $w_{\isix'-1,\isix'}^x$, $w_{\isix'+1,\isix'+2}^{\opvar(C,\ivar+1)}$ and $w_{\isix'+3,\isix'+4}^{\opvar(C,\ivar+2)}$.
Therefore the cliques that are adjacent to $w_{\isix'+1,\isix'+2}^{\opvar(C,\ivar+1)}$ and $w_{\isix'+3,\isix'+4}^{\opvar(C,\ivar+2)}$
are different from $Q_\isix^\ipart$, and these vertices do not belong to $K_\isix^\ipart$. We infer that if $s_{\icyc,\icop}^C \in K_\isix^\ipart$,
that is, $s_{\icyc,\icop}^C$ belongs to the same cluster as $w_{\isix'-1,\isix'}^x$, then $\ivar = \ivar(C) + \icop - 1$; equivalently,
$\icop = \ivar - \ivar(C) + 1$. Hence, $s_{\icyc,\icop}^C\in K_\isix^\ipart$ only if $C$ is satisfied by a variable in $\vars^\ipart_\isix$ and, providing this, for at most one choice of the indices $1 \leq \icyc,\icop \leq 3$. 
This concludes the proof of the claim.

We now count the number of editions in $G[\setX \cup \setS]$ as sketched
in the construction section.
The subgraph $G[\setX \cup \setS]$ contains $6n' + 27m'$ edges:
one $6$-cycle for each variable and three edges incident to each of the nine vertices
$s_{\icyc,\icop}^C$ for each clause $C$.
Each cluster $K_\isix^\ipart$ contains $n'/p$ vertices from
$\setX$ and $\frac{3m'}{2p}$ vertices from $\setS$.
If we deleted all edges in $G[\setX \cup \setS]$ and then added all the missing edges
in the clusters, we would make $\kcut + \kXSXS$ editions, due to the clusters being equal-sized. However, in this manner
we sometimes delete an edge and then introduce it again; thus, for each edge
of $G[\setX \cup \setS]$ that is contained in one cluster $K_\isix^\ipart$,
we should subtract $2$ in this counting scheme.

For each variable $x$, exactly three edges of the form $w_{\isix-1,\isix}^xw_{\isix,\isix+1}^x$
are contained in one cluster; this gives a total of $\kXX = 3n'$ saved edges.
For each clause $C$ each vertex $s_{\icyc,\icop}^C$ is assigned to a cluster with one of the vertices
$w_{2\icyc+2\ivar-3,2\icyc+2\ivar-2}^{\opvar(C,\ivar)}$, $1 \leq \ivar \leq 3$,
thus exactly one of the edges incident to $s_{\icyc,\icop}^C$ is contained
in one cluster. This sums up to $\kXS=9m'$ saved edges,
and we infer that the $6p$-cluster graph $G_0$ can be obtained from $G$ by exactly
$k' = \kQQ + \kQXS + \kcut + \kXSXS - 2\kXX - 2\kXS$ editions.
\end{proof}

\subsection{Soundness}\label{app:multi-2}

We need the following simple bound on the number of edges of a cluster graph.
\begin{lemma}\label{lem:even}
Let $a,b$ be positive integers and $H$ be a cluster graph with $ab$ vertices and at most $a$ clusters. Then $|E(H)| \geq a\binom{b}{2}$ and equality holds if and only if $H$ is an $a$-cluster graph and each cluster of $H$ has size exactly $b$.
\end{lemma}
\begin{proof}
It suffices to note that if not all clusters of $H$ are of size $b$, there is one of size at least $b+1$ and one of size at most $b-1$ or the number of clusters is less than $a$;
then, moving a vertex from the largest cluster of $H$ to a new or the smallest cluster strictly decreases the number of edges of $H$.
\end{proof}

We are now ready to show how to translate a $p'$-cluster graph $G_0$ with $p' \leq 6p$, $\ham(G_0,G) \leq k'$ into a satisfying assignment of the input formula $\Phi$.
\begin{lemma}
If there exists a $p'$-cluster graph $G_0$ with $V(G) = V(G_0)$, $p' \leq 6p$, $\ham(G,G_0) \leq k'$,
   then the formula $\Phi$ is satisfiable.
\end{lemma}
\begin{proof}
By Lemma \ref{lem:twins}, we may assume that each clique $Q_\isix^\ipart$ is contained in one cluster in $G_0$.
Let $F = E(G_0) \triangle E(G)$ be the editing set, $|F| \leq k'$.

Before we start, we present some intuition. The cluster graph $G_0$ may differ from the one constructed
in the completeness step in two significant ways, both leading to some savings in the edges of $G[\setX \cup \setS]$
that may not be included in $F$.
First, it may not be true that each cluster contains exactly one clique
$Q_\isix^\ipart$. However, since the number of cliques is at most $6p$, this may happen only if some clusters contain
more than one clique $Q_\isix^\ipart$, and we need to add $L^2$ edges to merge each pair of cliques that belong to the same cluster.
Second, a vertex $v \in \setX \cup \setS$ may not be contained in a cluster together with one of the cliques it is adjacent to.
However, as each such vertex needs to be separated from {\em{all}} its adjacent clusters (compared to {\em{all but one}} in the completeness step),
this costs us additional $L$ edges to remove.
The large constant in front of the definition of $L$ ensures us that in both these ways we pay more than we save on the edges of
$G[\setX \cup \setS]$.
%The main hardness of the proof is to show that the number of clusters of $G_0$
%is exactly $6p$, and each clique $Q_\isix^\ipart$ is contained in a different cluster.
%We prove that, if that is not the case, $|F|$ needs to be significantly
%larger than $k'$: intuitively, we need to add $L^2$ edges to merge two cliques
%$Q_\isix^\ipart$, but each such clique is incident to only $17n'/p < L$
%vertices in $\setX \cup \setS$; each such vertex gives only a gain of
%$L$ edges in $E(\setQ,\setX \cup \setS)$ and $\cO(1)$ edges of $G[\setX \cup \setS]$
%inside the cluster.
We now proceed to the formal argumentation.

We define the following quantities.
\begin{align*}
\lQQ &= |F \cap (\setQ \times \setQ)|, & \lQXS &= |F \cap E_G(\setQ,\setX \cup \setS)|, \\
\lXSXS &= |E(G_0(\setX \cup \setS))|, \\
\lXX &= |E_G(\setX,\setX) \cap E_{G_0}(\setX,\setX)|& \lXS &= |E_G(\setX,\setS) \cap E_{G_0}(\setX,\setS)|. \\
 \end{align*}
Recall that $\kcut = |E(G(\setX \cup \setS))| = 6n' + 27m'$. Similarly as in the completeness proof, we have that
$$|F| \geq \lQQ + \lQXS + \lXSXS + \kcut - 2\lXX - 2\lXS.$$
Indeed, $\lQQ$ and $\lQXS$ count (possibly not all) edges of $F$ that are incident to the vertices of $\setQ$.
The edges of $F \cap ((\setX \cup \setS) \times (\setX \cup \setS))$ are counted in an indirect way: each edge
of $G[\setX \cup \setS]$ is deleted ($\kcut$) and each edge of $G_0[\setX \cup \setS]$ is added ($\lXSXS$).
Then, the edges that are counted twice in this manner are subtracted ($\lXX$ and $\lXS$).

We say that a cluster is {\em{crowded}} if it contains at least two cliques $Q_\isix^\ipart$ and
{\em{proper}} if it contains exactly one clique $Q_\isix^\ipart$.
A clique $Q_\isix^\ipart$ that is contained in a crowded (proper) cluster
is called a {\em{crowded}} ({\em{proper}}) clique.
%A part $\ipart$, $1 \leq \ipart \leq p$ is {\em{crowded}}, if it contains at least one crowded clique,
%  and {\em{proper}} otherwise.

Let $a$ be the number of crowded cliques. Note that
$$\lQQ-\kQQ = |F \cap (\setQ \times \setQ)|-0 \geq aL^2/2,$$
as each vertex in a crowded clique needs to be connected to at least one other
crowded clique.

We say that a vertex $v \in \setX \cup \setS$ is {\em{attached}} to a clique $Q_\isix^\ipart$,
   if it is adjacent to all vertices of the clique in $G$. Moreover, we say that a vertex $v \in \setX \cup \setS$ is {\em{alone}} if it is contained in a cluster
in $G_0$ that does not contain any clique $v$ is attached to. Let $\XSalone$ be the number of alone vertices.

%We say that $v$ is {\em{attached}} to part $\ipart$,
%if it is attached to at least one clique $Q_\isix^\ipart$ for $1 \leq \isix \leq 6$.
%Note that $w_{\isix,\isix+1}^x \in \setX$ is attached to exactly one part $\ipart(x)$, whereas a vertex $s \in \setS$ can
%be attached to at most three parts.

Let us now count the number of vertices a fixed clique $Q_\isix^\ipart$ is attached to.
Recall that $|\vars^\ipart| = n'/p$. For each variable $x \in \vars^\ipart$ the clique $Q_\isix^\ipart$ is
attached to two vertices $w_{\isix-1,\isix}^x$ and $w_{\isix,\isix+1}^x$.
Moreover, each variable $x \in \vars^\ipart$ appears in exactly six clauses: thrice positively and thrice negatively.
For each such clause $C$, $Q_\isix^\ipart$ is attached to the vertex $s_{\icyc,2}^C$ for exactly one choice of the value
$1 \leq \icyc \leq 3$ and to the vertex $s_{\icyc,3}^C$ for exactly one choice of the value $1 \leq \icyc \leq 3$.
Moreover, if $x$ appears in $C$ positively and $\isix$ is odd, or if $x$ appears in $C$ negatively and $\isix$ is even,
then $Q_\isix^\ipart$ is attached to the vertex $s_{\icyc,1}^C$ for exactly one choice of the value $1 \leq \icyc \leq 3$.
We infer that the clique $Q_\isix^\ipart$ is attached to exactly fifteen vertices from $\setS$ for each variable $x \in \vars^\ipart$.
Therefore, there are exactly $17|\vars^\ipart| = 17n'/p$ vertices of $\setX \cup \setS$ attached to $Q_\isix^\ipart$: $2n'/p$ from $\setX$ and $15n'/p$ from $\setS$.

Take an arbitrary vertex $v \in \setX \cup \setS$ and assume that $v$ is attached to $b_v$ cliques,
and $a_v$ out of them are crowded. As $F$ needs to contain all edges of $G$ that connect $v$ with cliques
that belong to a different cluster than $v$, we infer that $|F \cap E_G(\{v\},\setQ)| \geq (b_v - \max(1,a_v))L$.
Moreover, if $v$ is alone, $|F \cap E_G(\{v\},\setQ)| \geq b_vL\geq 1\cdot L + (b_v - \max(1,a_v))L$. Hence 
\begin{eqnarray*}
\lQXS= |F \cap E_G(\setQ,\setX \cup \setS)| & \geq & \XSalone L + \sum_{v \in \setX \cup \setS} (b_v - \max(1,a_v))L \\
& \geq & \XSalone L + \sum_{v \in \setX \cup \setS} (b_v - 1)L - \sum_{v \in \setX \cup \setS} a_vL.
\end{eqnarray*}
Recall that $\sum_{v \in \setX \cup \setS} (b_v-1)L = \kQXS$. Therefore, using the fact that each clique
is attached to exactly $17n'/p$ vertices of $\setX \cup \setS$, we obtain that
$$\lQXS-\kQXS = |F \cap E_G(\setQ,\setX \cup \setS)| - \kQXS \geq \XSalone L - \sum_{v \in \setX \cup \setS} a_vL \geq \XSalone L - 17aLn'/p.$$

In $G_0$, the vertices of $\setX \cup \setS$ are split between $p' \leq 6p$ clusters and there are $6n' + 9m'$ of them.
By Lemma \ref{lem:even}, the minimum number of edges of $G_0[\setX \cup \setS]$ is attained when all clusters are of equal size
and the number of clusters is maximum possible. We infer that $\lXSXS \geq \kXSXS$.

We are left with $\lXX$ and $\lXS$.
Recall that $\kXX$ counts three edges out of each $6$-cycle constructed per variable of $\Phi'$, $|\kXX| = 3n'$,
whereas $\kXS$ counts one edge per each vertex $s_{\icyc,\icop}^C \in \setS$, $\kXS = 9m' = |\setS|$.

Consider a crowded cluster $K$ with $c>1$ crowded cliques. We say that $K$ {\em{interferes}} with a vertex $v \in \setX \cup \setS$
if $v$ is attached to a clique in $K$. As each clique is attached to exactly $17n'/p$ vertices of $\setX \cup \setS$, $2n'/p$ belonging to $\setX$ and $15n'/p$ to $\setS$, in total at most $2an'/p$ vertices of $\setX$ interfere with a crowded cluster and at most $15an'/p$ vertices of $\setS$.
%Let us now bound the number of vertices $K$ can interfere with.
%Recall that $|\vars^\ipart| = n'/p$ for $1 \leq \ipart \leq p$. For each variable $x \in \vars^\ipart$ there are six
%vertices $w_{\isix,\isix+1}^x$, $1 \leq \isix \leq 6$ attached to the part $\ipart$,
%         thus $K$ interferes with at most $6cn'/p$ vertices of $\setX$.
%Moreover, each variable $x \in \vars^\ipart$
%is contained in exactly six clauses, and each clause is represented by nine vertices $s_{\icyc,\icop}^C$, $1 \leq \icyc,\icop \leq 3$.
%Therefore, there are at most $54n'/p$ vertices of $\setX \cup \setS$ attached to a fixed part $\vars^\ipart$,
%          and $K$ interferes with at most $54cn'/p$ vertices of $\setS$.
%In total, there are at most $6an'/p$ vertices of $\setX$ and $54an'/p$ vertices of $\setS$ that interfere with at least one cluster of $G_0$.

Fix a variable $x \in \vars(\Phi')$.
If none of the vertices $w_{\isix,\isix+1}^x \in \setX$ interferes with any crowded cluster $K$, then
all the cliques $Q_{\isix'}^{\ipart(x)}$, $1 \leq \isix' \leq 6$, are proper cliques, each contained in a different cluster in $G_0$.
Moreover, if additionally no vertex $w_{\isix,\isix+1}^x$, $1 \leq \isix \leq 6$, is alone,
then in the $6$-cycle constructed for the variable $x$ at most three edges are not in $F$.
On the other hand, if some of the vertices $w_{\isix,\isix+1}^x \in \setX$ interfere with a crowded cluster $K$, or at least one of them is alone,
it may happen that all six edges of this $6$-cycle are contained in one cluster of $G_0$.
The total number of $6$-cycles that contain either alone vertices or vertices interfering with crowded clusters is bounded by $\XSalone + an'/p$, as every clique is attached to exactly $n'/p$ $6$-cycles. In $\kXX$ we counted three edges per a $6$-cycle, while in $\lXX$ we counted at most three edges per every $6$-cycles except $6$-cycles that either contain alone vertices or vertices attached to crowded cliques, for which we counted at most six edges. Hence, we infer that
$$\lXX - \kXX \leq 3(\XSalone + an'/p).$$

We claim that if a vertex $s_{\icyc,\icop}^C \in \setS$ (i) is not alone, and (ii) is not attached to a crowded clique, and (iii) is not adjacent to any alone vertex in $\setX$, then at most one edge from $E(\{s_{\icyc,\icop}^C\},\setX)$ may not be in $F$. Recall that $s_{\icyc,\icop}^C$ has exactly three neighbors in $\setX$, each of them attached to exactly two cliques and all these six cliques are pairwise distinct; moreover, $s_{\icyc,\icop}^C$ is attached only to these six cliques, if $\icyc=2,3$, or only to three out of these six, if $\icyc=1$. Observe that (i) and (ii) imply that $s_{\icyc,\icop}^C$ is in the same cluster as exactly one of the six cliques attached to his neighbors in $\setX$, so if it was in the same cluster as two of his neighbors in $\setX$, then at least one of them would be alone, contradicting (iii).
However, if at least one of (i), (ii) or (iii) is not satisfied, then all three edges incident to $s_{\icyc,\icop}^S$ may be contained in one cluster. As each vertex in $\setX$ is adjacent to at most $18$ vertices in $\setS$ (at most $3$ per every clause in which the variable is present), there are at most $18\XSalone$ vertices $s_{\icyc,\icop}^C$ that are alone
or adjacent to an alone vertex in $\setX$. Note also that the number of vertices of $\setS$ interfering with crowded clusters is bounded by $15an'/p$, as each of $a$ crowded cliques has exactly $15n'/p$ vertices of $\setS$ attached. Thus, we are able to bound the number of vertices of $\setS$ for which (i), (ii) or (iii) does not hold. As in $\kXS$ we counted one edge per every vertex of $\setS$, while in $\lXS$ we counted at most one edge per every vertex of $\setS$ except vertices not satisfying (i), (ii), or (iii), for which we counted at most three edges, we infer that
$$\lXS - \kXS \leq 2 (18\XSalone + 15an'/p).$$

Summing up all the bounds:
\begin{align*}
|F| - k' & \geq (\lQQ-\kQQ) + (\lQXS-\kQXS) + (\lXSXS-\kXSXS) \\
    &\qquad \qquad -  2(\lXX-\kXX) - 2(\lXS-\kXS) \\
    &\geq aL^2/2 + \XSalone L - 17aLn'/p + 0 - 6(\XSalone + an'/p) - 4(15\XSalone + 18an'/p) \\
    &\geq a+\XSalone \geq 0\\
\end{align*}
The second to last inequality follows from the choice of the value of $L$, $L = 1000\cdot\left(1+ \frac{n'}{p\varepsilon}\right)$; note that in particular $L \geq 1000$.

We infer that $a = 0$, that is, each clique $Q_\isix^\ipart$ is contained in a different cluster of $G_0$, and each cluster of $G_0$ contains exactly one such clique.
Moreover, $\XSalone = 0$, that is, each vertex $v \in \setX \cup \setS$ is contained in a cluster with at least one clique $v$ is attached to; as
all cliques are proper, $v$ is contained in a cluster with exactly one clique $v$ is attached to and $\lQXS = \kQXS$.

Recall that $|F \cap ((\setX \cup \setS) \times (\setX \cup \setS))| = \lXSXS + \kcut - 2\lXX - 2\lXS$. As each clique is now proper and no vertex is alone,
for each variable $x$ at most three edges out of the $6$-cycle $w_{\isix,\isix+1}^x$, $1 \leq \isix \leq 6$, are not in $F$, that is, $\lXX \leq \kXX$.
Moreover, for each vertex $s_{\icyc,\icop}^C \in \setS$,
the three neighbors of $s_{\icyc,\icop}^C$ are contained in different clusters and at most one edge incident to $s_{\icyc,\icop}^C$ is not in $F$, that is, $\lXS \leq \kXS$.
As $|F| \leq k'$, these inequalities are tight: exactly three edges out of each $6$-cycle are not in $F$, and
exactly one edge adjacent to a vertex in $\setS$ is not in $F$.

Consider an assignment $\phi'$ of $\vars(\Phi')$ that assigns $\phi'(x) = 1$ if the vertices
$w_{\isix,\isix+1}^x$, $1 \leq \isix \leq 6$ are contained in clusters with cliques
$Q_1^{\ipart(x)}$, $Q_3^{\ipart(x)}$, and $Q_5^{\ipart(x)}$
(i.e., the edges $w_{6,1}^xw_{1,2}^x$, $w_{2,3}^xw_{3,4}^x$ and $w_{4,5}^xw_{5,6}^x$ are not in $F$),
  and $\phi'(x) = 0$ otherwise (i.e., if the vertices $w_{\isix,\isix+1}^x$, $1 \leq \isix \leq 6$ are contained in clusters with cliques
$Q_2^{\ipart(x)}$, $Q_4^{\ipart(x)}$ and $Q_6^{\ipart(x)}$) --- a direct check shows that these are the only ways to save $3$ edges inside a $6$-cycle.
We claim that $\phi'$ satisfies $\Phi'$. Consider a clause $C$. The vertex $s_{1,1}^C$ is contained in a cluster with one of the three cliques it is attached to (as $\XSalone = 0$),
   say $Q_{\isix'}^\ipart$,
and with one of the three vertices of $\setX$ it is adjacent to, say $w_{\isix,\isix+1}^x$. Therefore $\ipart(x) = \ipart$,
    $w_{\isix,\isix+1}^x$ is contained in the same cluster as $Q_{\isix'}^\ipart$, and $\phi'(x)$ satisfies the clause $C$.
\end{proof}

\section{General clustering under ETH: proof of Theorem \ref{thm:eth}}\label{app:eth}

In this section we prove Theorem~\ref{thm:eth}, namely that the \clustering{} problem without
restriction on the number of clusters in the output
does not admit a $2^{o(k)} n^{\cO(1)}$ algorithm unless the Exponential Time Hypothesis fails.
\medskip
 
The following lemma provides a linear reduction from the problem of verifying satisfiability
of $3$-CNF formulas.

\begin{lemma}\label{lem:eth1}
There exists a polynomial-time algorithm that, given a $3$-CNF formula $\Phi$ with
$n$ variables and $m$ clauses, constructs a \clustering{} instance $(G,k)$
such that (i) $\Phi$ is satisfiable if and only if $(G,k)$ is a YES-instance,
and (ii) $|V(G)|+|E(G)|+k = \cO(n+m)$.
\end{lemma}

\begin{proof}
By standard arguments, we may assume that each clause of $\Phi$ consists of exactly three
literals with different variables and each variable appears at least twice:
at least once in a positive literal and at least once in a negative one.
Let $\textrm{Vars}(\Phi)$ denote the set of variables of $\Phi$.
For a variable $x$, let $s_x$ be the number of appearances of $x$ in the formula $\Phi$.
For a clause $C$ with variables $x$, $y$, and $z$, we denote by
$l_{x,C}$ the literal of $C$ that contains $x$ (i.e., $l_{x,C} = x$ or $l_{x,C} = \neg x$).

\medskip
\noindent
\textbf{Construction.} We construct a graph $G=(V,E)$ as follows.
First, for each variable $x$ we introduce a cycle $A_x$ of length $4s_x$.
For each clause $C$ where $x$ appears we assign four consecutive vertices
$a_{x,C}^j$, $1 \leq j \leq 4$ on the cycle $A_x$.
If the vertices assigned to a clause $C'$ follow the vertices assigned to a clause $C$
on the cycle $A_x$, we let $a_{x,C}^5=a_{x,C'}^1$.

Second, for each clause $C$ with variables $x$, $y$, and $z$ we introduce
a gadget with $6$ vertices $V_C = \{p_x, p_y, p_z, q_x, q_y, q_z\}$
with all inner edges except for $q_xq_y$, $q_yq_z$, and $q_zq_x$
(see Figure \ref{fig:eth1}).
If $l_{x,C} = x$ then we connect $q_x$ to the vertices $a_{x,C}^1$ and
$a_{x,C}^2$, and if $l_{x,C} = \neg x$, we connect $q_x$ to $a_{x,C}^2$
and $a_{x,C}^3$. We proceed analogously for variables $y$ and $z$ in the clause $C$.
We set $k = 8m + 2\sum_{x \in \textrm{Vars}(\Phi)} s_x = 14m$.
This finishes the construction of the \clustering{} instance $(G,k)$. Clearly
$|V(G)| + |E(G)| + k = O(n+m)$. We now prove that $(G,k)$ is a YES-instance
if and only if $\Phi$ is satisfiable.

\begin{figure}
\centering
\includegraphics{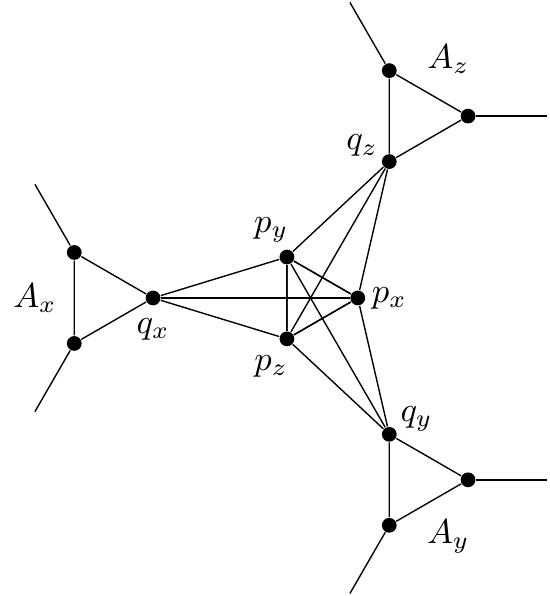}
\caption{The gadget for a clause $C$ with variables $x$, $y$ and $z$.}
  \label{fig:eth1}
  \end{figure}

\medskip
\noindent
\textbf{Completeness.} Assume that $\Phi$ is satisfiable, and let $\phi$
be a satisfying assignment for $\Phi$. We construct a set $F \subseteq V \times V$
as follows. First, for each variable $x$
we take into $F$ the edges $a_{x,C}^2a_{x,C}^3$, $a_{x,C}^4a_{x,C}^5$
for each clause $C$ if $\phi(x)$ is true and the edges $a_{x,C}^1a_{x,C}^2$, $a_{x,C}^3a_{x,C}^4$
for each clause $C$ otherwise. 
Second, let $C$ be a clause of $\Phi$ with variables $x$, $y$, and $z$
and, without loss of generality, assume that the literal $l_{x,C}$ satisfies $C$
in the assignment $\phi$.
For such a clause $C$ we add to $F$ eight elements:
the edges $q_xp_x$, $q_xp_y$, $q_xp_z$,
the four edges that connect $q_y$ and $q_z$ to the cycles $A_y$ and $A_z$,
and the non-edge $q_yq_z$.

Clearly $|F| = \sum_{x \in \textrm{Vars}(\Phi)} 2s_x + 8m = k$. We now verify that $G \triangle F$
is a cluster graph. For each cycle $A_x$, the removal of the edges in $F$ results
in splitting the cycle into $2s_x$ two-vertex clusters. For each clause $C$ with
variables $x$, $y$, $z$, satisfied by the literal $l_{x,C}$ in the
assignment $\phi$, the vertices $p_x$, $p_y$, $p_z$, $q_y$, and $q_z$ form a $5$-vertex
cluster. Moreover, since $l_{x,C}$ is true in $\phi$, the edge that connects
the two neighbors of $q_x$ on the cycle $A_x$ is not in $F$, thus $q_x$ and these
two neighbors form a three-vertex cluster.

\medskip
\noindent
\textbf{Soundness.} Let $F$ be a minimum size feasible solution to the \clustering{} instance
$(G,k)$. %such that $|F|$ is minimum possible.
By Lemma \ref{lem:twins}, for each clause $C$ with variables $x$, $y$, and $z$,
the vertices $p_x$, $p_y$, and $p_z$ are contained in a single cluster in $G \triangle F$.
Denote the vertex set of this cluster by $Z_C$.
We choose $F$ (with minimum possible cardinality) such that the number of clusters
$Z_C$ that are contained in the vertex set $V_C$ is maximum possible.

Informally, we are going to show that the solution $F$ needs to look almost like
the one constructed in the proof of completeness. The crucial observation is
that if we want to create a six-vertex cluster $Z_C=V_C$
then we need to put nine (instead of eight) elements in $F$ that are incident to $V_C$.
Let us now proceed to the formal arguments.

Fix a variable $x$ and let $F_x = F \cap (V(A_x) \times V(A_x))$.
We claim that $|F_x| \geq 2s_x$ and, moreover, if $|F| = 2s_x$
then $F_x$ consists of every second edge of the cycle $A_x$.
Note that $A_x \triangle F_x$ is a cluster graph; assume that there are $\gamma$
clusters in $A_x \triangle F_x$ with sizes $\alpha_j$ for $1 \leq j \leq \gamma$.
If $\gamma = 1$ then, as $s_x \geq 2$,
\[
|F_x| = |\alpha_1| = \binom{4s_x}{2} - 4s_x = 8s_x^2 - 6s_x > 2s_x.
\]
Otherwise, in a cluster with $\alpha_j$ vertices we need to add at least
$\binom{\alpha_j}{2} - (\alpha_j - 1)$ edges and remove at least two
edges of $A_x$ leaving the cluster. Using~$\sum a_j=4s_x$, we infer that
\[
|F_x| \geq \gamma + \sum_{j=1}^\gamma \binom{\alpha_j}{2} - (\alpha_j - 1) = \frac{1}{2} \sum_{j=1}^\gamma \alpha_j^2 - 3\alpha_j + 4 = 2s_x + \frac{1}{2} \sum_{j=1}^\gamma (\alpha_j-2)^2.
\]
Thus, $|F_x| \geq 2s_x$ and $|F_x| = 2s_x$ only if for all $1 \leq j \leq \gamma$ we have
$\alpha_j=2$ and in each two-vertex cluster of $A_x \triangle F_x$, $F_x$ does not contain
the edge in this cluster and contains two edges of $A_x$ that leave this cluster. This
situation occurs only if $F_x$ consists of every second edge of the cycle $A_x$.

We now focus on a gadget for some clause $C$ with variables $x$, $y$, and $z$.
Let $F_C = F \cap (V_C\times (V_C \cup V(A_x) \cup V(A_y) \cup V(A_z)))$.
We claim that $|F_C| \geq 8$ and there are very limited ways in which we can obtain $|F_C| = 8$.

Recall that the vertices $p_x$, $p_y$, and $p_z$ are contained in a single cluster in $G \triangle F$
with vertex set $Z_C$.
We now distinguish subcases, depending on how many of the vertices
$q_x$, $q_y$, and $q_z$ are in $Z_C$.

If $q_x,q_y,q_z \notin Z_C$,
then $\{p_x,p_y,p_z\} \times \{q_x,q_y,q_x\} \subseteq F_C$ and $|F_C| \geq 9$.

If $q_x \in Z_C$, but $q_y,q_z \notin Z_C$,
then $\{p_x,p_y,p_z\} \times \{q_y,q_z\} \subseteq F_C$.
If there is a vertex $v \in Z_C \setminus V_C$,
then $F$ needs to contain three elements $vp_x$, $vp_y$, and $vp_z$.
In this case $F'$ constructed from $F$ by replacing all elements incident to $\{q_x,p_x,p_y,p_z\}$
with all eight edges of $G$ incident to this set is a feasible solution to $(G,k)$
of size smaller than $F$, a contradiction to the assumption of the minimality of $F$.
Thus, $Z_C = \{q_x,p_x,p_y,p_z\}$, and $F_C$ contains the eight edges of $G$ incident to 
$Z_C$.

If $q_x,q_y \in Z_C$ but $q_z \notin Z_C$, then $q_zp_x,q_zp_y,q_zp_z,q_xq_y \in F_C$.
If there is a vertex $v \in Z_C \setminus V_C$, then $F_C$ contains the three edges
$vp_x,vp_y,vp_z$ and at least one of the edges $vq_x$, $vq_y$.
In this case $F'$ constructed from $F$ by replacing
all elements incident to $\{p_x,p_y,p_z,q_x,q_y\}$ with all seven edges of $G$ incident
to this set and a non-edge $q_xq_y$ is a feasible solution to $(G,k)$ of size not greater than $F$,
with $Z_C \subseteq V_C$, a contradiction
to the choice of $F$. Thus $Z_C = \{p_x,p_y,p_z,q_x,q_y\}$ and
$F_C$ contains all seven edges incident to $Z_C$ and the non-edge $q_xq_y$.

In the last case, $V_C \subseteq Z_C$, and $q_xq_y,q_yq_z,q_zq_x \in F_C$.
There are six edges connecting $V_C$ and $V(A_x) \cup V(A_y) \cup V(A_z)$ in $G$, and all these edges
are incident to different vertices of $V(A_x) \cup V(A_y) \cup V(A_z)$. Let $uv$ be one of these six edges,
$u \in V_C$, $v \notin V_C$. If $v \in Z_C$ then $F$ contains five non-edges
connecting $v$ to $V_C \setminus \{u\}$. Otherwise, if $v \notin Z_C$,
$F$ contains the edge $uv$. We infer that $F_C$ contains at least six elements
that have exactly one endpoint in $V_C$  and $|F_C| \geq 9$.

We now note that the sets $F_C$ for all clauses $C$ and the sets $F_x$ for all variables $x$
are pairwise disjoint. Recall that $|F_x| \geq 2s_x$ for any variable $x$ and
$|F_C| \geq 8$ for any clause $C$. As $|F| \leq 14m = 8m + \sum_x 2s_x$, we infer
that $|F_x| = 2s_x$ for any variable $x$, $|F_C| = 8$ for any clause $C$
and $F$ contains no elements that are not in any set $F_x$ or $F_C$.

As $|F_x| = 2s_x$ for each variable $x$, the set $F_x$ consists of every second edge
of the cycle $A_x$. We construct an assignment $\phi$ as follows: $\phi(x)$ is true if for all clauses $C$ where $x$ appears
we have $a_{x,C}^2a_{x,C}^3,a_{x,C}^4a_{x,C}^5 \in F$ and $\phi(x)$ is false
if $a_{x,C}^1a_{x,C}^2,a_{x,C}^3a_{x,C}^4 \in F$. We claim that
$\phi$ satisfies $\Phi$. Consider a clause $C$ with variables $x$, $y$, and $z$.
As $|F_C| = 8$, by the analysis above one of two situations occur: $|Z_C| = 4$,
   say $Z_C = \{p_x,p_y,p_z,q_x\}$, or $|Z_C| = 5$, say $Z_C = \{p_x,p_y,p_z,q_x,q_y\}$.
In both cases, $F_C$ consists only of all edges of $G$ that connect $Z_C$ with $V(G) \setminus Z_C$
and the non-edges of $G[Z_C]$. Thus, in both cases the two edges that connect $q_z$
with the cycle $A_z$ are not in $F$. Thus, the two neighbors of $q_z$ on the cycle $A_z$
are connected by an edge not in $F$, and $\phi(z)$ satisfies the clause $C$.
\end{proof}

%\begin{theorem}\label{thm:eth1}
%  \clustering{} is not solvable in time $2^{o(k)}n^{\cO(1)}$ time unless Exponential
%  Time Hypothesis fails.
% \end{theorem} 

Lemma~\ref{lem:eth1} directly implies the proof of Theorem~\ref{thm:eth}

\begin{proof}[Proof of Theorem~\ref{thm:eth}]
A subexponential algorithm for \clustering{}, combined with the reduction
shown in Lemma~\ref{lem:eth1}, would give a subexponential (in the number of variables
and clauses) algorithm for verifying satisfiability of $3$-CNF formulas.
An existence of such algorithm is known to violate ETH \cite{ImpagliazzoPZ01}.
\end{proof}

We note that the graph constructed in the proof of Lemma \ref{lem:eth1} is of maximum
degree $5$. Thus our reduction shows that sparse instances of \clustering{}
where in the output the clusters are of constant size are hard.

\section{Conclusion and open questions}\label{sec:conclusions}
We gave an algorithm that solves  \textsc{$p$-Cluster Editing} in time $\cO(2^{\cO(\sqrt{pk})} +n+m)$ and complemented it by a multivariate lower bound, which shows that the running time of our algorithm is asymptotically tight for all $p$ sublinear in $k$.

In our multivariate lower bound it is crucial that the cliques and clusters are arranged in groups of six. However, the drawback of this construction is that Theorem \ref{thm:multivariate-reduction} settles the time complexity of \pclustering problem only for $p \geq 6$ (Corollary \ref{cor:fixed-p}). It does not seem unreasonable that, for example, the $2$-\textsc{Cluster Editing} problem, already NP-complete \cite{ShamirST04}, may have enough structure to allow an algorithm with running time $\cO(2^{o(\sqrt{k})}+n+m)$. Can we construct such an algorithm or refute its existence under ETH?

Secondly, we would like to point out an interesting link between the subexponential parameterized complexity of the problem and its approximability. When the number of clusters drops from linear to sublinear in $k$, we obtain a phase transition in parameterized complexity from exponential to subexponential. As far as approximation is concerned, we know that bounding the number of clusters by a constant allows us to construct a PTAS~\cite{GiotisG06}, whereas the general problem is APX-hard~\cite{CharikarGW05j}. The mutual drop of the parameterized complexity of a problem --- from exponential to subexponential --- and of approximability --- from APX-hardness to admitting a PTAS --- can be also observed for many hard problems when the input is constrained by additional topological bounds, for instance excluding a fixed pattern as a minor~\cite{DemaineFHT05jacm,DemaineH05a,FLRSsoda2011}. It is therefore an interesting question, whether \pclustering also admits a PTAS when the number of clusters is bounded by a non-constant, yet sublinear function of $k$, for instance $p=\sqrt{k}$.

\paragraph*{Acknowledgements} We thank Christian Komusiewicz for pointing us
to the recent results on \clustering{} \cite{bocker:iwoca,komusiewicz:sofsem} and his thesis
\cite{komusiewicz:thesis}. Moreover, we thank P\aa l Gr\o n\aa s Drange, M. S. Ramanujan and Saket Saurabh for helpful discussions.

\bibliographystyle{plain}
\bibliography{clustering}

\begin{thebibliography}{10}

\bibitem{AilonCN08}
Nir Ailon, Moses Charikar, and Alantha Newman.
\newblock Aggregating inconsistent information: Ranking and clustering.
\newblock {\em Journal of the ACM}, 55(5):23:1--23:27, 2008.

\bibitem{AlonLS09}
Noga Alon, Daniel Lokshtanov, and Saket Saurabh.
\newblock Fast {FAST}.
\newblock In {\em Proceedings of the 36th International Colloquium on Automata,
  Languages and Programming (ICALP 2009)}, volume 5555 of {\em {Lecture Notes
  in Computer Science}}, pages 49--58. Springer, 2009.

\bibitem{AlonMMN05}
Noga Alon, Konstantin Makarychev, Yury Makarychev, and Assaf Naor.
\newblock Quadratic forms on graphs.
\newblock In {\em Proceedings of the 37th ACM Symposium on Theory of Computing
  (STOC 2005)}, pages 486--493. ACM, 2005.

\bibitem{AroraBKSH05}
Sanjeev Arora, Eli Berger, Elad Hazan, Guy Kindler, and Muli Safra.
\newblock On non-approximability for quadratic programs.
\newblock In {\em Proceedings of the 46th Annual IEEE Symposium on Foundations
  of Computer Science (FOCS 2005)}, pages 206--215. IEEE Computer Society,
  2005.

\bibitem{Bansal04}
Nikhil Bansal, Avrim Blum, and Shuchi Chawla.
\newblock Correlation clustering.
\newblock {\em Machine Learning}, 56:89--113, 2004.

\bibitem{Ben-DorSY99}
Amir Ben-Dor, Ron Shamir, and Zohar Yakhini.
\newblock Clustering gene expression patterns.
\newblock {\em Journal of Computational Biology}, 6(3/4):281--297, 1999.

\bibitem{bocker:iwoca}
Sebastian B{\"o}cker.
\newblock A golden ratio parameterized algorithm for cluster editing.
\newblock {\em Journal of Discrete Algorithms}, 16:79--89, 2012.

\bibitem{BockerBBT08}
Sebastian B{\"o}cker, Sebastian Briesemeister, Quang Bao~Anh Bui, and Anke
  Tru{\ss}.
\newblock A fixed-parameter approach for weighted cluster editing.
\newblock In {\em Proceedings of the 6th Asia-Pacific Bioinformatics Conference
  (APBC 2008)}, volume~6 of {\em Advances in Bioinformatics and Computational
  Biology}, pages 211--220, 2008.

\bibitem{BockerBK11}
Sebastian B{\"o}cker, Sebastian Briesemeister, and Gunnar~W. Klau.
\newblock Exact algorithms for cluster editing: Evaluation and experiments.
\newblock {\em Algorithmica}, 60(2):316--334, 2011.

\bibitem{BockerD11}
Sebastian B{\"o}cker and Peter Damaschke.
\newblock Even faster parameterized cluster deletion and cluster editing.
\newblock {\em Information Processing Letters}, 111(14):717--721, 2011.

\bibitem{BodlaenderFHMPR10}
Hans~L. Bodlaender, Michael~R. Fellows, Pinar Heggernes, Federico Mancini,
  Charis Papadopoulos, and Frances~A. Rosamond.
\newblock Clustering with partial information.
\newblock {\em Theoretical Computer Science}, 411(7-9):1202--1211, 2010.

\bibitem{CaoC10}
Yixin Cao and Jianer Chen.
\newblock Cluster editing: Kernelization based on edge cuts.
\newblock {\em Algorithmica}, 64(1):152--169, 2012.

\bibitem{CharikarGW05j}
Moses Charikar, Venkatesan Guruswami, and Anthony Wirth.
\newblock Clustering with qualitative information.
\newblock {\em Journal of Computer and System Sciences}, 71(3):360--383, 2005.

\bibitem{CharikarW04}
Moses Charikar and Anthony Wirth.
\newblock Maximizing quadratic programs: Extending {G}rothendieck's inequality.
\newblock In {\em Proceedings of the 45th Symposium on Foundations of Computer
  Science (FOCS 2004)}, pages 54--60. IEEE Computer Society, 2004.

\bibitem{ChenM10}
Jianer Chen and Jie Meng.
\newblock A $2k$ kernel for the cluster editing problem.
\newblock {\em Journal of Computer and System Sciences}, 78(1):211--220, 2012.

\bibitem{Damaschke10}
Peter Damaschke.
\newblock Fixed-parameter enumerability of cluster editing and related
  problems.
\newblock {\em Theory of Computing Systems}, 46(2):261--283, 2010.

\bibitem{DemaineFHT05jacm}
Erik~D. Demaine, Fedor~V. Fomin, MohammadTaghi Hajiaghayi, and Dimitrios~M.
  Thilikos.
\newblock Subexponential parameterized algorithms on graphs of bounded genus
  and {$H$}-minor-free graphs.
\newblock {\em Journal of the ACM}, 52(6):866--893, 2005.

\bibitem{DemaineH05a}
Erik~D. Demaine and MohammadTaghi Hajiaghayi.
\newblock Bidimensionality: {N}ew connections between {FPT} algorithms and
  {PTASs}.
\newblock In {\em Proceedings of the 16th Symposium on Discrete Algorithms
  (SODA 2005)}, pages 590--601, 2005.

\bibitem{DowneyF99}
R.~G. Downey and M.~R. Fellows.
\newblock {\em Parameterized complexity}.
\newblock Springer-Verlag, New York, 1999.

\bibitem{FellowsGKNU11}
Michael~R. Fellows, Jiong Guo, Christian Komusiewicz, Rolf Niedermeier, and
  Johannes Uhlmann.
\newblock Graph-based data clustering with overlaps.
\newblock {\em Discrete Optimization}, 8(1):2--17, 2011.

\bibitem{FlumGrohebook}
J{\"o}rg Flum and Martin Grohe.
\newblock {\em Parameterized Complexity Theory}.
\newblock Texts in Theoretical Computer Science. An EATCS Series.
  Springer-Verlag, Berlin, 2006.

\bibitem{FLRSsoda2011}
Fedor~V. Fomin, Daniel Lokshtanov, Venkatesh Raman, and Saket Saurabh.
\newblock Bidimensionality and {EPTAS}.
\newblock In {\em Proceedings of the 22nd Symposium on Discrete Algorithms
  (SODA 2011)}, pages 748--759. SIAM, 2011.

\bibitem{FominV126soda}
Fedor~V. Fomin and Yngve Vilanger.
\newblock Subexponential parameterized algorithm for minimum fill-in.
\newblock In {\em Proceedings of the 23rd Symposium on Discrete Algorithms
  (SODA 2012)}, pages 1737--1746. SIAM, 2012.

\bibitem{GiotisG06}
Ioannis Giotis and Venkatesan Guruswami.
\newblock Correlation clustering with a fixed number of clusters.
\newblock In {\em Proceedings of the 17th Symposium on Discrete Algorithms
  (SODA 2006)}, pages 1167--1176. ACM Press, 2006.

\bibitem{GrammGHN05}
Jens Gramm, Jiong Guo, Falk H{\"u}ffner, and Rolf Niedermeier.
\newblock Graph-modeled data clustering: Exact algorithms for clique
  generation.
\newblock {\em Theory of Computing Systems}, 38(4):373--392, 2005.

\bibitem{Guo09}
Jiong Guo.
\newblock A more effective linear kernelization for cluster editing.
\newblock {\em Theoretical Computer Science}, 410(8-10):718--726, 2009.

\bibitem{GuoKKU11}
Jiong Guo, Iyad~A. Kanj, Christian Komusiewicz, and Johannes Uhlmann.
\newblock Editing graphs into disjoint unions of dense clusters.
\newblock {\em Algorithmica}, 61(4):949--970, 2011.

\bibitem{GuoKNU10}
Jiong Guo, Christian Komusiewicz, Rolf Niedermeier, and Johannes Uhlmann.
\newblock A more relaxed model for graph-based data clustering: s-plex cluster
  editing.
\newblock {\em SIAM Journal of Discrete Mathematics}, 24(4):1662--1683, 2010.

\bibitem{ImpagliazzoPZ01}
Russell Impagliazzo, Ramamohan Paturi, and Francis Zane.
\newblock Which problems have strongly exponential complexity?
\newblock {\em Journal of Computer and System Sciences}, 63(4):512--530, 2001.

\bibitem{komusiewicz:thesis}
Christian Komusiewicz.
\newblock {\em Parameterized Algorithmics for Network Analysis: Clustering \&
  Querying}.
\newblock PhD thesis, Technische Universit\"{a}t Berlin, 2011.
\newblock Available at
  \url{http://fpt.akt.tu-berlin.de/publications/diss-komusiewicz.pdf}.

\bibitem{komusiewicz:sofsem}
Christian Komusiewicz and Johannes Uhlmann.
\newblock Alternative parameterizations for cluster editing.
\newblock In {\em Proceedings of the 37th International Conference on Current
  Trends in Theory and Practice of Computer Science (SOFSEM 2011)}, volume 6543
  of {\em Lecture Notes in Computer Science}, pages 344--355. Springer, 2011.

\bibitem{marx:future}
D{\'a}niel Marx.
\newblock What's next? future directions in parameterized complexity.
\newblock In {\em The Multivariate Algorithmic Revolution and Beyond}, volume
  7370 of {\em Lecture Notes in Computer Science}, pages 469--496. Springer,
  2012.

\bibitem{ProttiSS09}
F{\'a}bio Protti, Maise~Dantas da~Silva, and Jayme~Luiz Szwarcfiter.
\newblock Applying modular decomposition to parameterized cluster editing
  problems.
\newblock {\em Theory of Computing Systems}, 44(1):91--104, 2009.

\bibitem{ShamirST04}
Ron Shamir, Roded Sharan, and Dekel Tsur.
\newblock Cluster graph modification problems.
\newblock {\em Discrete Applied Mathematics}, 144(1-2):173--182, 2004.

\end{thebibliography}

\end{document}